\documentclass[reqno,final]{amsart}
\usepackage{natbib}  
\usepackage{fancyhdr} 
\usepackage{color} 
\usepackage{hyperref} 
\usepackage{graphicx} 
\usepackage[ulem=normalem]{changes}

\usepackage{pstricks}
\usepackage{amssymb}
\usepackage{amsfonts}
\usepackage{amsmath}

\definecolor{aleacolor}{rgb}{0.16,0.59,0.78}

\hypersetup{
breaklinks,
colorlinks=true,
linkcolor=aleacolor,
urlcolor=aleacolor,
citecolor=aleacolor}

\renewcommand{\cite}{\citet}

\theoremstyle{plain}
\newtheorem{theorem}{Theorem}[section]                                          
                          
\newtheorem{lemma}[theorem]{Lemma}
\newtheorem{corollary}[theorem]{Corollary}

\newtheorem{condition}{Hypothesis}

\theoremstyle{definition}
\newtheorem{definition}[theorem]{Definition}
\theoremstyle{remark}
\newtheorem{remark}{Remark}[section]
\makeatletter \@addtoreset{equation}{section} \makeatother

\newcommand{\mlt}[1]{\left\lceil#1\right\rceil}

\begin{document}

\title[Stochastic Schr\"odinger equations]
{Stochastic Schr\"odinger equations and applications to Ehrenfest-type theorems}

\author{F. Fagnola}
\author{C.M. Mora}

\address{Dipartimento di Matematica, Politecnico di Milano 
\newline
Piazza Leonardo Da Vinci 32, I-20133 \newline
Milano, Italy.}

\address{CI$^2$MA and Departamento de Ingenier\'{\i}a Matem\'{a}tica, Universidad de Concepci\'on \newline
Barrio Universitario, Avenida     Esteban Iturra s/n,
Casilla 160-C \newline
Concepci\'on, Chile.}

\email{franco.fagnola@polimi.it, cmora@ing-mat.udec.cl}
\urladdr{
\url{http://www.mate.polimi.it/viste/pagina_personale/pagina_personale.php?id=186&lg=en},
\url{http://www.ing-mat.udec.cl/~cmora}
}
\thanks{This research was partially supported by FONDECYT Grants 7090088, 1070686 and 1110787.
FF was partially supported by the PRIN 2010-2011 programme "Evolution differential 
problems: deterministic and stochastic approaches and their interactions".
CMM was partially supported by BASAL Grants PFB-03 and FBO-16 as well as by PBCT-ACT 13 project.}



\begin{abstract}
 We study stochastic evolution equations describing the dynamics 
of open quantum systems.
First, using resolvent approximations, 
we obtain a sufficient condition for  regularity of solutions to linear stochastic Schr\"{o}dinger equations driven by cylindrical Brownian motions 
applying  to many physical systems.
Then, we establish 
well-posedness and norm conservation property of a wide class of open quantum systems described in position representation.
Moreover,
we prove Ehrenfest-type theorems that describe the evolution of the mean value of quantum observables in open systems.
Finally, 
we give a new  criterion for  the existence and uniqueness of weak 
solutions to non-linear stochastic Schr\"{o}dinger equations.
We apply our results to physical systems such as fluctuating ion traps
and quantum measurement processes of position. 
 \end{abstract}

\maketitle



\section{Introduction}

Stochastic Schr\"odinger equations are frequently used to describe 
quantum measurement processes (see, e.g., \cite{Barchielli,WisemanMilburn2009})
and, in general, 
quantum systems that are sensitive to the environment influence 
(see, e.g., \cite{GardinerZoller2004,Carmichael2008}).
Moreover,
non-linear stochastic Schr\"odinger equations are becoming an established tool
for  numerical simulation of the evolution of open quantum systems 
(see, e.g.,  \cite{Breuer,Percival}).
This motivates the study of 
mathematical properties of stochastic Schr\"odinger equations allowing us to obtain information 
on physical phenomena.
In this research direction,
we first investigate regularity of solutions to  linear and non-linear stochastic Schr\"odinger equations arising in the study of quantum systems with continuous variables,
namely having  $L^{2}\left( \mathbb{R}^{d} ,\mathbb{C}\right) $ as state space.
Then,
we prove a version of Ehrenfest's theorem for open quantum systems.
As a concrete physical application, 
we deduce rigorously the  linear heating in a Paul trap.

In Section \ref{sec:MainRes},
we first focus on open quantum systems described by
the linear stochastic evolution equation in a complex separable Hilbert space 
$\left(\mathfrak{h},\left\langle \cdot,\cdot\right\rangle \right) $:
\begin{equation}
\label{eq:SSE}
X_{t}\left( \xi \right) 
= \xi +\int_{0}^{t}G \left( s \right) X_{s}\left( \xi \right) ds 
+ \sum_{\ell=1}^{\infty }\int_{0}^{t}
L_\ell \left( s \right) X_{s}\left( \xi \right) dW_{s}^{\ell}, 
\end{equation}
see, e.g., 
\cite{Barchielli,Barchielli1,BassiDurrKolb2010,Belavkin1,Breuer,Gehm1998,Gough,Grotz,Halliwell,Schneider,Singh2007} and the references therein.
The driving noise 
$\left( W^{\ell} \right)_{\ell \ge 1}$ 
is a sequence of real valued independent Wiener processes on a filtered complete probability 
space $\left( \Omega ,\mathfrak{F},
\left(\mathfrak{F}_{t}\right) _{t\geq 0},\mathbb{P}\right) $,
the solution $X$ is a pathwise continuous adapted stochastic processes 
taking values in $\mathfrak{h}$,
$\xi\in L^2(\Omega,\mathbb{P})$,
and $\left( G \left( t \right)  \right)_{t \geq 0}, \left( L_{\ell}
\left( t \right)  \right)_{t \geq 0}$ are 
given families of linear operators on $\mathfrak{h}$ satisfying  
\begin{equation}
\label{eq:G(t)}
G \left( t \right)  =
-iH \left( t \right) - \frac{1}{2} \sum_{\ell=1}^{\infty}
L_{\ell} \left( t \right)  ^*L_{\ell} \left( t \right) 
\end{equation}
on suitable common domain with $H\left( t \right) $ symmetric operator.
The relation (\ref{eq:G(t)}) is a necessary  condition for mean norm square conservation of  $X_{t}\left( \xi \right) $, an important physical property
that must hold in the application to open quantum systems.

In Subsection \ref{subsec:2.1}
we establish a sufficient condition for regularity of solutions to (\ref{eq:SSE}),
closely adapted to its special structure. 
Regular solutions are essentially solutions with finite 
energy, indeed, regularity of  $X_{t}\left( \xi \right) $ is characterized  
through $ \mathbb{E}\left\Vert C X_t(\xi)\right\Vert^2 < \infty$ 
for suitable non-negative operators $C$ on $\mathfrak{h}$, with a 
domain contained in the domains of  $G(t)$ and $L_\ell(t)$,
allowing us to control unboundedness of these operators.
Taking inspiration from resolvent approximation 
methods developed in  \cite{FagnolaWills}, 
we strengthen results of 
 \cite{MoraMC2004} and \cite{MoraReIDAQP,MoraReAAP} 
and improve their applicability to open quantum systems with 
infinite dimensional state space in coordinate representation
(see Section \ref{sec:LinearSSE1} for a review of previous works).
Moreover, we prove that regularity of $X$ implies the mean norm square conservation property, 
namely  
$\mathbb{E} \left\Vert X_t(\xi)\right\Vert^2=\mathbb{E}\left\Vert \xi\right\Vert^2$ 
for all $t\ge 0$.

In Subsection  \ref{sec:nonlinear_SSE}, 
we report our 
careful verification that existence and uniqueness of the regular solution to  
(\ref{eq:SSE}) yields existence and uniqueness of the regular solution to 
\begin{equation}
 \label{eq:nlSSE}
 Y_{t} 
 = 
 Y_{0}
 +
  \int_{0}^{t} G \left(s, Y_{s} \right) ds 
 + \sum_{\ell=1}^{\infty} 
 \int_{0}^{t} \left(L_{\ell} \left(s \right)   Y_{s}    
 - \Re \left\langle  Y_{s}, 
 L_{\ell}  \left(s \right)   Y_{s}  \right\rangle  Y_{s} \right)
  dW_{s}^{\ell}
\end{equation}
with 
$$
G\left(s, y\right)
=
G  \left(s \right)  y  
+
\sum_{\ell =1}^{\infty}\left(\Re \left\langle y,L_{\ell}  \left(s \right)   y \right\rangle
L_{\ell}  \left(s \right)  y - \frac{1}{2}  \Re^{2}\left\langle y,L_{\ell}  \left(s \right) y \right\rangle y \right) .
$$
We thus get from Subsection \ref{subsec:2.1}
a sufficient condition for well-posedness of (\ref{eq:nlSSE}).
This non-linear stochastic Schr\"odinger equation 
is a fundamental tool for modeling the dynamics of states in quantum measurement 
processes
(see, e.g., \cite{Barchielli1,Barchielli,BassiDurrKolb2010,Belavkin1,Breuer,Gough}),
as well as numerical simulation of the evolution 
of mean values of quantum observables 
(see, e.g., \cite{Breuer,MoraAAP2005,Percival}),
which are represented  by 
$ \mathbb{E} \langle Y_{t} , A Y_{t}  \rangle$.

Mathematics of closed quantum systems is well established, on the contrary,
 only a few papers  deal with open quantum systems
whose state space $\mathfrak{h}$ contains, among its components, $L^{2}\left( \mathbb{R}^{d} ,\mathbb{C}\right)$
(see, e.g., \cite{BassiDurrKolb2010,ChebFagn2,Kolo,Gough,MoraReAAP,MoraAP} and references therein).
However,
important physical phenomena are realistically described by
open quantum systems involving continuous variables such as position
(see, e.g., \cite{DAgostaDiVentra2008,Gough,Halliwell,HarocheRaimond2006,WisemanMilburn2009}).
This motivates Section \ref{subsec:oqs-coord} where
we use our general results as the starting point for investigating
well-posedness and norm
conservation property of physical systems described in position 
representation with Hilbert space 
$\mathfrak{h} =  L^{2}\left( \mathbb{R}^{d} ,\mathbb{C}\right) $, Hamiltonian 
\begin{equation}\label{eq:Hex}
 H(t) = - \alpha \Delta + i\sum_{j=1}^d \left(A^j(t,\cdot)\partial_j 
+\partial_j A^j(t,\cdot)\right) + V(t, \cdot)
\end{equation}
and noise coefficients
\begin{equation}
 \label{eq:Lex}
 L_{\ell} \left(t \right)   =
\left\{ 
\begin{array}{ll}
\sum_{j=1}^d \sigma_{\ell j}\left(t,\cdot\right)\partial_j + \eta_{\ell} \left(t , \cdot \right), 
& \text{if } 1 \leq \ell \leq m
\\ 
0, & \text{if }  \ell>m
\end{array}
\right. ,
\end{equation}
where  $ t \geq 0$,  $m\in\mathbb{N}$,
$\alpha$ is a non-negative real constant,
$\partial_j$ denotes the partial derivative with respect to the 
$j^{\hbox{\rm th}}$-coordinate,
$V,A^j : [0,+\infty[\times \mathbb{R}^d \rightarrow \mathbb{R}$ 
and 
$ \sigma_{\ell j}$, $\eta_{\ell}  : [0,+\infty[\times \mathbb{R}^d \rightarrow \mathbb{C}$
are measurable smooth functions.
We thus include in our study concrete physical situations like:
continuous measurements of 
position  \cite{Bassi2009,DurrHinKolb,Gough,Kolo},
atoms in interaction with polarized lasers  \cite{Singh2007},
quantum systems in fluctuating traps  \cite{Grotz,Schneider}
and collisions of heavy-ions  \cite{Alicki,ChebFagn2}.
The main difficulties in the study of stochastic partial 
differential equations (\ref{eq:SSE}) and (\ref{eq:nlSSE})
with Hamiltonian (\ref{eq:Hex}) and noise operators  (\ref{eq:Lex})
lies in the unboundedness of partial derivatives $\partial_j$ in the noise coefficients 
as well as in  the magnetic fields terms 
$A^j(t,\cdot)\partial_j  +\partial_j A^j(t,\cdot)$, 
a possible linear growth of functions $\eta_\ell$
and the possible quadratic behavior of the potential  $V$; 
solving (\ref{eq:SSE}) and (\ref{eq:nlSSE}) 
we must cope with all of them at the same time.
We overcome these difficulties by using
the reference operator 
$C = - {\Delta}+ \left| x \right| ^{2} $,
together with non-trivial algebraic and analytic manipulations.

In Section \ref{sec:Ehrenfest}, 
we derive rigorously Ehrenfest-type theorems for open quantum systems.
Indeed, assuming that (\ref{eq:SSE}) has a unique $C$-regular solution,
we prove, roughly speaking, that  the mean value of a $C$-bounded observable $A$ satisfies:
\begin{eqnarray}
\label{eq:I.1}
 \mathbb{E} \left\langle X_{t} \left( \xi \right) , A X_{t} \left( \xi \right) \right\rangle
& = &  
\mathbb{E} \left\langle  \xi  , A  \xi  \right\rangle
+
\int_0^t 
\mathbb{E} \left\langle A^* X_{s} \left( \xi \right) , G  \left( s \right)  X_{s} \left( \xi \right) \right\rangle
 ds
\\
\nonumber
&&
+ \int_0^t 
\mathbb{E} \left\langle G  \left( s \right)  X_{s} \left( \xi \right) , A  X_{s} \left( \xi \right) \right\rangle 
ds
\\
\nonumber
&&
+
\int_0^t \left(
\sum_{\ell = 1}^{\infty}
\mathbb{E} \left\langle  L_{\ell}  \left( s \right)  X_{s} \left( \xi \right) , A L_{\ell}  \left( s \right)  X_{s} \left( \xi \right) \right\rangle
\right) ds .
\end{eqnarray}
States of quantum systems are described 
by density operators, i.e., positive operators on $\mathfrak{h}$ with unit trace.
Under, for instance, the Born-Markov approximation,
the density operator at time $t$ is given (in Dirac notation) by 
$$
 \rho_t = \mathbb{E} \left|X_{t} \left( \xi \right)\right\rangle
\left\langle X_{t} \left( \xi \right)\right|\ 
$$
whenever the initial density operator is 
$ \mathbb{E} \left| \xi \right\rangle \left\langle \xi \right|$
(see, e.g., \cite{Barchielli,Breuer,MoraAP,Percival}).
Hence 
the mean value of a $C$-bounded observable $A$ is well-defined by $ tr \left(  \rho_t  A \right) $,
which is equal to 
$ \mathbb{E} \langle X_{t} \left( \xi \right), A X_{t} \left( \xi \right) \rangle$ 
(see, e.g.,  \cite{MoraAP}),
and (\ref{eq:I.1}) becomes 
\begin{equation}
\label{eq:3}
  \frac{d}{dt} tr \left(  \rho_t  A \right) 
=
tr \left(  \rho_t  \left(
-i \left[ A,  H \left( t \right)  \right]
+ \frac{1}{2} L_{\ell} \left( t \right)^{*} \left[ A,  L_{\ell} \left( t \right)  \right]
+ \frac{1}{2} \left[L_{\ell} \left( t \right)^{*} , A   \right] L_{\ell} \left( t \right) 
\right)  \right),
\end{equation}
where 
$\left[ \cdot,  \cdot \right] $ stands for  the commutator between two operators
and $tr \left(  \cdot \right)$ denotes the trace operation.

Ehrenfest-type theorems describe the rate 
of change of mean values of quantum observables.
In the physical literature on open quantum systems,
the generalized Ehrenfest equations (\ref{eq:I.1}) and (\ref{eq:3}) have been used, for example, 
to demonstrate  connections between quantum and classical mechanics (see, e.g., \cite{Percival}),
and to estimate the behavior of the expected value of important quantum observables
\cite{Breuer,EnglertMorigi2002,Halliwell,HupinLacroix2010,Salmilehto2012}.
Nevertheless,
(\ref{eq:I.1}) and (\ref{eq:3}) 
have not been rigorously examined
from the mathematical viewpoint.
This motivates Section \ref{sec:Ehrenfest}
 where we present 
the first, 
to the best of our knowledge,
rigorous proof of
the Ehrenfest equations (\ref{eq:I.1}) and (\ref{eq:3}) 
for open quantum systems with infinite-dimensional state space $\mathfrak{h}$.
We would like to point out here that
Ehrenfest-type theorems for closed quantum systems 
have been recently proved by
\cite{Friesecke2009,Friesecke2010}; 
our results also generalize this work.

In Section \ref{sec:Ehrenfest}, we also introduce 
sufficient conditions for  validity of  (\ref{eq:I.1})  
and (\ref{eq:3}) applied to the system 
with Hamiltonian (\ref{eq:Hex}) and noise operators  (\ref{eq:Lex}).
This, together with Section  \ref{subsec:oqs-coord},
provides a sound framework for studying open quantum systems 
in coordinate  representation with smooth potentials. 

As a concrete physical application we consider ions traps
(see, e.g., \cite{Wineland1998} for a description).
Quadrupole ion traps were initially developed by 
Hans Georg Dehmelt and Wolfgang Paul 
who  were awarded
the Nobel Prize in Physics for this work
having a great impact in quantum information.
Experiments show that these traps lose coherence,
because the  coupling with the environment is relatively strong
(see, e.g.,  \cite{Grotz,Leibfried2006,Wineland1998} and references therein).
This drastically reduces   life times of  trapped atoms.
Here,
we prove rigorously the linear heating in a model of a Paul trap
whenever the initial density operator is regular enough,
providing  a
mathematically rigorous presentation of the arguments given by  \cite{Schneider}.

\vskip 1truecm

\subsection{Notation}
\label{subsec:not}

In this article,  $\left(\mathfrak{h},\left\langle \cdot,\cdot\right\rangle \right) $ is a separable complex Hilbert space whose scalar product $\left\langle \cdot,\cdot \right\rangle $ is linear in the second variable and anti-linear in the first one. 
We write $\mathcal{D}\left(A\right)$ for the domain of $A$, whenever $A$ is a linear operator in $\mathfrak{h}$. 
If $\mathfrak{X}$, $\mathfrak{Z}$ are normed spaces,
then we denote by $\mathfrak{L}\left( \mathfrak{X},\mathfrak{Z}\right) $ the set of all bounded operators from $\mathfrak{X}$ to $\mathfrak{Z}$
and we define $\mathfrak{L}\left( \mathfrak{X}\right) = \mathfrak{L}\left( \mathfrak{X},\mathfrak{X}\right) $. 
We set $  \left[ A,  B \right] = AB - BA$
when $A,B$  are  operators   in  $\mathfrak{h}$.
By  $ \mathcal{B} \left( \mathfrak{Y} \right)$ we mean  the set of all Borel set of the topological space $ \mathfrak{Y}$.

Suppose that $C$ is a self-adjoint positive operator in $\mathfrak{h}$. 
For  any $x,y\in \mathcal{D}\left( C\right) $
we define the graph scalar product 
$\left\langle x,y\right\rangle_{C}=\left\langle x,y\right\rangle +\left\langle Cx,Cy\right\rangle $
 and the graph norm 
 $ \left\Vert x\right\Vert _{C}=
\sqrt{\left\langle x,x\right\rangle _{C}}$.
We use the symbol $L^{2}\left( \mathbb{P},\mathfrak{h}\right) $ to denote  the set of all square integrable random variables from $\left( \Omega ,\mathfrak{F},\mathbb{P}\right)$ to $ \left( \mathfrak{h},\mathfrak{B}\left( \mathfrak{h}\right) \right)$.
Moreover,  $L_{C}^{2}\left( \mathbb{P},\mathfrak{h}\right) $  stands for the set of all $\xi \in L^{2}\left( \mathbb{P},\mathfrak{h}\right) $ such that $\xi \in \mathcal{D}\left( C\right) $ a.s. and $\mathbb{E} \left\Vert \xi \right\Vert _{C}^{2}<\infty $. 
We define $\pi _{C}:\mathfrak{h\rightarrow h}$  by 
$\pi_C(x)=x$ if $x\in \mathcal{D}\left( C\right) $ and 
$\pi_C(x)=0$ if $x\notin \mathcal{D}\left( C\right) $. 

In case $g : \mathbb{R}^{n} \mapsto \mathbb{C}$ is Borel measurable,
$\mlt{g}$ stands for the multiplication operator in 
$L^{2}\left( \mathbb{R}^{n} ,\mathbb{C}\right) $ 
given by $f \mapsto g f$.
We abbreviate $\mlt{g}$ to $g$ when no confusion can arise. 
We denote by $C^{k}\left( \mathbb{R}^{d} ,\mathbb{K}\right) $ 
with $\mathbb{K} = \mathbb{R}, \mathbb{C}$,
 the set of all functions from $\mathbb{R}^{d}$ to $\mathbb{K}$  
 whose  partial derivatives up to order $k$ are  with continuous.
Moreover, $C^{\infty}_{c}\left( \mathbb{R}^{d} ,\mathbb{C}\right) $ 
is the set of all functions of $C^{\infty} \left( \mathbb{R}^{d},  \mathbb{C}\right) $ having compact support.
If $f:\mathbb{R}^{d} \mapsto \mathbb{C}$, 
then $\partial_k f$ denotes the partial derivative of $f$ 
with respect to its $k$-th argument, 
$\nabla f$ stands for the gradient 
of $f$ and $ {\Delta}f$ is the Laplacian of $f$.

In what follows, the letter $K$ denotes generic constants.  
We will write $K \left( \cdot \right)$ for different  non-decreasing 
non-negative functions on the interval $\left[ 0, \infty \right[$ when 
no confusion is possible. 

\section{Stochastic Schr\"{o}dinger equations} 
\label{sec:MainRes}

\subsection{Linear stochastic Schr\"{o}dinger equation}
\label{subsec:2.1}

\subsubsection{Previous works}
\label{sec:LinearSSE1}

In the autonomous case,  \cite{Holevo} obtained  the 
existence and uniqueness of the weak (topological) solution 
to (\ref{eq:SSE}) whenever $G$ is the infinitesimal generator 
of a  contraction semigroup. A drawback of  such weak solutions
is that they may not preserve  the mean value of  
$\left\Vert X_t(\xi)\right\Vert^2$ (see, e.g., \cite{Holevo}).
\cite{Rozovskii} proved the existence and uniqueness 
of  variational solutions for a  class dissipative 
linear stochastic evolution equations on real Hilbert spaces, 
where the regularity of $X_t(\xi)$ is essentially characterized
through a strictly positive operator $C$.
In particular,
approximating 
$G \left( s \right) $ by $G \left( s \right) 
- \epsilon C^2$ in (\ref{eq:SSE}),
 \cite{Rozovskii} 
obtained a solution of (\ref{eq:SSE})  as a limit of 
solutions to coercive  stochastic evolution equations that 
are treated using  the Galerkin method.
This indirect proof makes it difficult to address 
some properties of the SSEs  as time-global 
estimates  needed for establishing  the existence of regular 
invariant measures for  (\ref{eq:nlSSE}), and time-local estimates appearing in the numerical solution of  (\ref{eq:SSE})  and (\ref{eq:nlSSE}) (see, e.g.,   \cite{MoraMC2004}).
Using Galerkin approximations,
\cite{GreckschSAA}
proved the existence and uniqueness of variational solutions to
\begin{equation}
 \label{eq:4}
 d X_t 
=
\left(  i \left( - H_0 X_t +  f \left( t ,  X_t \right) \right) \right) dt
+
ig \left( t ,  X_t \right)  dW_{t} ,
\end{equation}
where 
$W$ is a cylindrical Brownian motion with values in a separable real Hilbert space,
$f$, $g$ are locally Lipschitz functions
and 
$- H_0$ is a coercive operator with discrete spectrum.
These conditions are strong  in case (\ref{eq:4}) becomes linear. 

Applying directly the Galerkin method,  together with a priori 
estimates of the graph norm of the approximating solutions 
with respect to the reference positive operator $C$, 
 \cite{MoraMC2004} and \cite{MoraReIDAQP} proved  that (\ref{eq:SSE})
 has a unique  strong regular solution, in the autonomous case.
The assumptions of   \cite{MoraMC2004} and \cite{MoraReIDAQP} 
include the existence of an orthonormal basis 
$\left( e_n \right)_n$ of $\left(\mathfrak{h},\left\langle \cdot,\cdot\right\rangle \right) $ that satisfies, for instance,
$\sup_{n\in\mathbb{Z}_{+}}\left\Vert CP_{n}x\right\Vert
\leq\left\Vert C x\right\Vert $ for all $x$ belonging to the 
domain of $C$, where $P_{n}$ is the orthogonal projection 
of  $\mathfrak{h}$ over the linear manifold spanned by 
$e_{0},\ldots e_{n}$ and summability of the series 
$\sum_{\ell}\left\Vert L_\ell ^* e_n\right\Vert^2$ together 
with some domain hypotheses on the adjoint $G^*$ of $G$.  
Summability of the series 
$\sum_{\ell}\left\Vert L_\ell ^* e_n\right\Vert^2$, in particular, 
is a strong mathematical requirement that may not hold even when the operators 
$G$ and $L_\ell$ are bounded. 
In Section \ref{subsec:LinearSSE},
we prove the well-posedness of (\ref{eq:SSE}),
as well as the regularity of its solution,
under hypotheses that do not involve 
the orthogonal basis  $\left( e_n \right)_n$, the 
summability condition  and technical hypotheses on adjoints 
of  $G$ and $L_\ell$ (that now are also time-dependent).
Then,
we obtain stronger results with simplified proofs
and wider range of applications.

Finally,
the non-commutative version of  (\ref{eq:SSE}) has been treated 
using resolvent approximations and a priori estimates by \cite{FagnolaWills}.

\subsubsection{Main results}
\label{subsec:LinearSSE}

We start by making precise the notion of  strong regular solution to  (\ref{eq:SSE}).

\begin{condition}
\label{hyp:L-G-C-domain}
Let $C$ be a self-adjoint positive operator in $\mathfrak{h}$ such that:
\begin{itemize}
 \item[(H1.1)] 
 For any $\ell\ge 1$ and $t \geq 0$,
 $\mathcal{D}\left(C \right) 
\subset \mathcal{D}\left( L_\ell \left( t \right) \right)$ and 
 $L_{\ell}  \left( \cdot \right)  \circ \pi _{C}$ is measurable 
as a function from 
 $\left( \left[ 0 , \infty \right[ \times \mathfrak{h}, 
\mathcal{B}\left(  \left[ 0 , 
\infty \right[ \times \mathfrak{h} \right) \right) $ to 
$\left(\mathfrak{h}, \mathcal{B} \left( \mathfrak{h} \right)\right)$.
  
 \item[(H1.2)] 
 For all $t \geq 0$,  
 $\mathcal{D}\left( C\right) \subset \mathcal{D}\left( G \left( t \right) \right)$. 
 Moreover,  
 $$
 G \left( \cdot \right) \circ \pi _{C}
 : 
 \left( \left[ 0 , \infty \right[ \times \mathfrak{h}, 
\mathcal{B}\left(  \left[ 0 , 
\infty \right[ \times \mathfrak{h} \right) \right) 
\rightarrow 
\left(\mathfrak{h}, \mathcal{B} \left( \mathfrak{h} \right)  \right)
 $$
 is measurable.
\end{itemize}
 \end{condition}

\begin{definition} \label{def:regular-sol}
Let Hypothesis \ref{hyp:L-G-C-domain} hold. 
Assume that $\mathbb{I}$ is 
either $\left[ 0,\infty \right[ $ or the interval $\left[ 0,T\right] $, 
with $T\in \mathbb{R}_{+}$. 
An $\mathfrak{h}$-valued adapted process $\left( X_{t}\left( \xi \right) \right) _{t \in \mathbb{I}}$  
with continuous sample paths is called strong 
$C$-solution of (\ref{eq:SSE}) on  $\mathbb{I}$ with initial datum $\xi$ if and only if, for all $t\in \mathbb{I}$:
\begin{itemize}
\item $\mathbb{E}\left\Vert X_{t}\left( \xi \right) \right\Vert ^{2} 
\leq \mathbb{E}\left\Vert \xi \right\Vert ^{2}$, $X_{t}\left( \xi \right) \in \mathcal{D}\left( C\right) $ a.s. and 
$
\sup_{s\in \left[ 0,t\right] }\mathbb{E}\left\Vert C X_{s}\left( \xi \right)  \right\Vert ^{2} < \infty 
$.
\item
$
X_{t}\left( \xi \right) 
=\xi 
+\int_{0}^{t}G \left( s \right) \pi _{C}\left( X_{s}\left( \xi \right) \right) ds
+\sum_{\ell=1}^{\infty }
\int_{0}^{t}L_\ell \left( s \right) \pi _{C}\left( X_{s}\left( \xi \right) \right) dW_{s}^\ell
$  \hspace{0.1cm} $\mathbb{P}$-a.s. 
\end{itemize}
\end{definition}

The lemma below guarantees that 
Hypothesis \ref{hyp:L-G-C-domain} is valid
in many physical models.

\begin{lemma}
 \label{lema:Measurability:g}
 Consider the self-adjoint positive operator 
 $C: \mathcal{D}\left( C\right) \subset \mathfrak{h} \rightarrow \mathfrak{h} $.
 Suppose that each family of linear operators
 $\left( G \left( t \right)  \right)_{t \geq 0} and  \left( L_{\ell} \left( t \right)  \right)_{t \geq 0}$,  
 with $\ell \in \mathbb{N}$, can be written  as 
 $$
 \left( \sum_{k=1}^{n} f_k \left( t \right) \Phi_k \right)_{t \geq 0} ,
 $$
 where 
 $f_1, \ldots, f_n :  \left( \left[ 0 , \infty \right[ , \mathcal{B}\left(  \left[ 0 , \infty \right[  \right) \right) 
\rightarrow 
\left(\mathbb{C}, \mathcal{B} \left( \mathbb{C} \right)  \right) 
 $
 are measurable 
 and
 $\Phi_1, \ldots, \Phi_n$
 belong to $\mathfrak{L}\left( \left( \mathcal{D}\left( C\right), \left\Vert \cdot \right\Vert _{C}\right)  ,\mathfrak{h}\right)$.
 Then Hypothesis \ref{hyp:L-G-C-domain} is fulfilled.
\end{lemma}

\begin{proof}
Deferred to Subsection \ref{subsec:lemaMeasurability:g}.
\end{proof}

\begin{remark}
Assume that \eqref{eq:SSE} is autonomous, i.e.,  
$G \left( t \right)$ and $L_{\ell} \left( t \right)$  do not depend on $t$.
From Lemma \ref{lema:Measurability:g} we have that Hypothesis \ref{hyp:L-G-C-domain} holds in case  
$G, L_{\ell}  \in \mathfrak{L}\left( \left( \mathcal{D}\left( C\right), \left\Vert \cdot \right\Vert _{C}\right)  ,\mathfrak{h}\right)$,
where
$C$ is a self-adjoint positive operator on $\mathfrak{h}$.
\end{remark}

The following theorem provides  a new general sufficient condition for 
the existence and uniqueness of strong $C$-solutions to (\ref{eq:SSE}).

\begin{condition}\label{hyp:CF-inequality}
Let $C$ satisfy Hypothesis \ref{hyp:L-G-C-domain}.  
In addition assume that:
\begin{itemize}

\item[(H2.1)]  For all $t \geq 0$ and 
$x \in \mathcal{D}\left( C\right)$,
$\left\Vert  G  \left( t \right) x  \right\Vert^{2} \leq K \left( t \right) \left\Vert  x  \right\Vert_{C}^{2}$ .

\item[(H2.2)] For every natural number $\ell$ there exists a non-decreasing  function 
$K_{\ell}$ on  $ \left[ 0 , \infty \right[ $ satisfying
$
\left\Vert  L_{\ell}  \left( t \right) x  \right\Vert^{2} 
\leq K_{\ell} \left( t \right) \left\Vert  x  \right\Vert_{C}^{2}
$
for all $x \in \mathcal{D}\left( C\right)$ and $t \geq 0$.

\item[(H2.3)] There exists a non-decreasing non-negative function $\alpha$ such that  
\[
2\Re\left\langle C^{2} x, G \left( t \right) x\right\rangle 
+\sum_{\ell=1}^{\infty }\left\Vert C L_{\ell} \left( t \right) x \right\Vert ^{2}
\leq \alpha \left( t \right)  \left\Vert x\right\Vert_{C}^{2}
\]
for all $t \geq 0$
and any $x$ belonging to a core $\mathfrak{D}_{1}$ of $C^{2}$.

\item[(H2.4)] There exists a core $\mathfrak{D}_{2}$ of $C$ such 
that 
$
2\Re\left\langle  x, G \left( t \right) x\right\rangle 
+\sum_{\ell=1}^{\infty }\left\Vert  L_{\ell} \left( t \right) x \right\Vert ^{2}\leq 0
$
for all $x \in \mathfrak{D}_{2}$ and $t \geq 0$.

\end{itemize}
\end{condition}

\begin{theorem}
\label{teorema1}
Let Hypothesis \ref{hyp:CF-inequality} hold 
and 
assume  that $\xi \in L_{C}^{2}\left( \mathbb{P},\mathfrak{h}\right) $ is
$\mathfrak{F}_{0}$-mea\-surable. 
Then (\ref{eq:SSE}) has a unique strong $C$-solution 
$\left( X_{t}\left( \xi \right) \right) _{t \geq0}$ with initial datum $\xi$. 
Moreover, 
\[ 
\mathbb{E} \left\Vert C  X_{t}\left( \xi \right) \right\Vert ^{2} 
\leq 
\exp \left( t \alpha \left( t \right) \right) \left( \mathbb{E} \left\Vert C  \xi  \right\Vert ^{2} 
+ t \alpha\left( t \right) \mathbb{E} \left\Vert  \xi  \right\Vert ^{2} \right).
\]
\end{theorem}

\begin{proof}
Deferred to Subsection \ref{subsec:existence}.
\end{proof}

\begin{remark}
Under the assumptions and notation of Theorem \ref{teorema1},
we can prove  the Markov property of  $X_{t}\left( \xi \right)$
by  techniques of well-posed martingale problems
(see, e.g., \cite{MoraReAAP}).
\end{remark}

The next lemma provides an equivalent formulation of Condition H2.3,
stated in terms of random variables.

\begin{lemma}
\label{lema:EquivH2.3}
Suppose that  $C$ is a self-adjoint positive operator in $\mathfrak{h}$ 
such that 
$G \left( t \right)$ and $C L_{\ell} \left( t \right)$
belong to 
$\mathfrak{L}\left( \left( \mathcal{D}\left( C^2\right), \left\Vert \cdot \right\Vert _{C^2}\right)  ,\mathfrak{h}\right)$
for all $t \geq 0$ and $\ell \in \mathbb{N}$. 
We define $\mathcal{L}^{+} \left( t, x  \right)$ to be the positive part of  
$
 2\Re\left\langle C^{2} x, G \left( t \right) x\right\rangle 
+\sum_{\ell=1}^{\infty }\left\Vert C L_{\ell} \left( t \right) x \right\Vert ^{2} 
$
whenever $t \geq 0$ and $x \in  \mathcal{D}\left( C^2\right)$.
Assume that $\mathfrak{D}_{1}$ is a a core of $C^{2}$.
Then, 
Condition H2.3 holds if and only if:
 \begin{itemize}
  \item[(H2.3')] 
  For all $\zeta \in  L_{C}^{2}\left( \mathbb{P},\mathfrak{h}\right)$ satisfying 
  $\zeta \in \mathfrak{D}_{1}$ and $\left\Vert \zeta \right\Vert = 1$,
  the function 
 $$
 t \mapsto \mathbb{E} \left( \mathcal{L}^{+} \left( t, \zeta \right)  \right)
 $$
 is bounded on any interval $\left[ 0 , T \right]$,
 with $T>0$.
\end{itemize}
 \end{lemma}

\begin{proof}
Deferred to Subsection \ref{subsec:lemaEquivH2.3}.
\end{proof}

Under Hypothesis \ref{hyp:CF-inequality} and Condition H3.1 below,
we can obtain the mean norm square conservation of  $X_t(\xi)$,
a crucial physical property of the quantum systems. 

\begin{condition}
\label{hyp:formal-conservativity} 
Let Hypothesis \ref{hyp:L-G-C-domain} hold
together with Condition H2.1.
Suppose that:
\begin{itemize}
\item[(H3.1)]  
For all  $t \geq 0$ and $x\in \mathcal{D}\left( C\right) $,
$
 2\Re\left\langle x, G \left( t \right) x\right\rangle 
+ \sum_{\ell=1}^{\infty }\left\Vert L_\ell \left( t \right) x\right\Vert ^{2} = 0.
$

\item[(H3.2)] 
For any initial datum $\xi$ belonging to $L_{C}^{2}\left( \mathbb{P},\mathfrak{h}\right)$,
(\ref{eq:SSE}) has a unique strong $C$-solution on any bounded interval.

\end{itemize}
\end{condition}

\begin{theorem}
\label{teorema2}
Assume that Hypothesis \ref{hyp:formal-conservativity} holds, 
together with $\xi\in L_{C}^{2}\left(\mathbb{P};\mathfrak{h}\right)$.
Then  $\left(  \left\Vert  X_{t}\left( \xi \right)  \right\Vert ^{2}\right) _{t} $ 
is a martingale. In particular $ \mathbb{E}\left\Vert X_t(\xi)\right\Vert^2 =\left\Vert \xi\right\Vert^2$ 
for all $t \geq 0$.
\end{theorem}

\begin{proof}
Deferred to Subsection \ref{subsec:teorema2}.
\end{proof}

\begin{remark}
Condition H3.1 is a quadratic form version of (\ref{eq:G(t)}). It 
arises from physical situations where we can expect that the solutions of the quantum master equations have trace $1$ at any time. Nevertheless,  (\ref{eq:G(t)}) is not a sufficient 
condition for a minimal quantum dynamical semigroup to 
be identity preserving 
(see, e.g., \cite{Fagnola}).
\end{remark}

\begin{remark}
Hypothesis \ref{hyp:CF-inequality}, together with Condition H3.1,  
constitutes a generalized version of  non-explosion criteria 
used to guarantee the conservation of the probability mass 
of minimal quantum dynamical semigroups
(see, e.g., \cite{ChebFagn2,ChebGarQue98,Fagnola}). 
This can be verified in a wide range of applications.
\end{remark}

\begin{remark}
The operator $C$  in Theorem \ref{teorema1} plays the  role  of superharmonic (or excessive) functions
in the Lyapunov condition for non-explosion of classical minimal Markov processes.
For simplicity,
suppose that $G \left( t \right)$ and $L_{\ell} \left( t \right)$ are time-independent. 
In this case 
Condition H2.3 of Hypothesis \ref{hyp:CF-inequality}
formally reads as
$$
\mathcal{L} \left( C^2 \right)
\leq
\alpha \left( C^2 + I \right) ,
$$
where $\alpha > 0$
and 
$
\mathcal{L} \left( X \right)
: =
G^{*} X + X G + \sum_{\ell = 1}^{\infty} L_{\ell}^{*} X L_{\ell}
$. Here 
$\mathcal{L}  \left( X \right)$
represents the infinitesimal generator of the Markov process $X_t$
applied to the function $x \mapsto \langle x , X x \rangle$.
Actually, we can choose $C$ satisfying
$
\mathcal{L} \left( C^2 \right)
\leq
\alpha C^2
$,
hence
$$
\frac{ d } {d t} \exp \left( - \alpha t \right) C^2  
+ 
\mathcal{L} \left(  \exp \left( - \alpha t \right) C^2   \right) \leq 0 .
$$
Thus,
$
\phi \left( t, x \right)
:=
\exp \left( - \alpha t \right) \left\|  C  x \right\|^2
$
is, roughly speaking, an  $\alpha$-excessive function.
Therefore,
applying formally  It\^o's formula we obtain that 
$ \exp \left( - \alpha t \right) \left\|  C  X_t \right\|^2 $ is a supermartingale.
Heuristically, 
$\phi$ helps us to prove that $X_t$ 
does not escape from the domain of $C$,
like the existence of superharmonic functions prevents finite explosion times
in classical Markov processes. 
\end{remark}

\subsection{Non-linear stochastic Schr\"{o}dinger equations}
\label{sec:nonlinear_SSE}

Using the linear stochastic Schr\"{o}\-dinger equation (\ref{eq:SSE}),
\cite{Barchielli1} construct  a weak probabilistic solution of (\ref{eq:nlSSE})
provided that  $G$ and $L_{1},L_{2},\ldots$ are bounded operators;
they actually considered driven noises with jumps  in place of some $W^\ell$.
In the case where $\mathfrak{h}$ is finite-dimensional 
and at most a finite number of $L_{k}$
are different from $0$,
the existence and uniqueness of the strong solution of (\ref{eq:nlSSE})
was obtained in Lemma 5 of  \cite{MoraAAP2005} 
by classical methods for stochastic differential equations with locally Lipschitz coefficients,
see also \cite{Barchielli,Pellegrini2008,Pellegrini2010}.

\cite{Gatarek1991} established the existence and  pathwise uniqueness of solutions of (\ref{eq:nlSSE})
in the following two examples: 
\begin{itemize}
 \item $H=0$, $L_{1}$ self-adjoint and  $L_{\ell}=0$ for all $\ell \geq2$.
 
 \item Let $\mathfrak{h} = L^{2}\left(  \mathbb{R},\mathbb{C}\right)$.
Choose $H=- \Delta $, $L_{1} f \left( x \right) = x f \left( x \right)$,
and 
 $L_{2}= L_{3}= \cdots = 0$.
\end{itemize}
To handle the  uniqueness property,
\cite{Gatarek1991}  used strongly that $L_{1}$ is a self-adjoint operator.
\cite{MoraReAAP} obtained  the existence and weak uniqueness of regular solutions to (\ref{eq:nlSSE})
under the assumptions of \cite{MoraReIDAQP},
which were discussed in Section \ref{sec:LinearSSE1}.
In the preparation of this paper,
we verified that 
applying the same arguments of the proof of Theorem 1 of  \cite{MoraReAAP} 
we can prove Theorem \ref{teorema7},
asserting 
the existence and uniqueness of solutions to 
the non-linear stochastic Schr\"{o}dinger equation  (\ref{eq:nlSSE})
under Hypothesis \ref{hyp:formal-conservativity}.
We thus get that
Theorem \ref{teorema1} provides a sufficient condition
for the existence and uniqueness of  weak (in the probabilistic sense) regular solution to (\ref{eq:nlSSE}).

\begin{definition}
Let $C$ satisfy Hypothesis \ref{hyp:L-G-C-domain}. 
Suppose that $\mathbb{I}$ is either $\left[ 0,+\infty\right[ $ 
or $\left[ 0,r\right] $ with  $r\in\mathbb{R}_+ $. 
We say that 
$\left( \Omega ,\mathfrak{F},\left( \mathfrak{F}_{t}\right) _{t\in\mathbb{I}},\mathbb{Q},
\left( Y_{t}\right) _{t\in\mathbb{I}},\left( W_{t}^{\ell}\right) _{t\in\mathbb{I}}^{\ell\in\mathbb{N}}\right) $
is a solution of class $C$ of (\ref{eq:nlSSE}) with initial 
distribution $\theta$ on the interval $\mathbb{I}$ if and only if:
\begin{itemize}
\item $\left( W^{\ell} \right)_{\ell\ge 1}$ is a sequence of real valued independent Brownian motions on
the filtered complete probability space 
$\left( \Omega,\mathfrak{F},\left( 
\mathfrak{F}_{t}\right) _{t\in\mathbb{I}},\mathbb{Q}\right) $.

\item $\left( Y_{t}\right) _{t\in\mathbb{I}}$ is an 
$\mathfrak{h}$-valued process with continuous sample paths 
such that the law of $Y_{0}$ coincides with $\theta$ 
and $\mathbb{Q}\left( \left\Vert Y_{t}\right\Vert =1 \; 
\mathrm{for}\; \mathrm{all}\; t\in\mathbb{I}\right) =1$.
Moreover, for every $t\in\mathbb{I}$, 
$Y_{t}\in \mathcal{D}\left( C\right)$ $\mathbb{Q}$-$a.s.$ 
and $\sup_{s\in\left[ 0,t\right] }\mathbb{E}_{\mathbb{Q}}\left\Vert CY_{s}\right\Vert ^{2}<\infty$.

\item $\mathbb{Q}$-$a.s.$, for all $t\in\mathbb{I}$,
\begin{eqnarray*}
 Y_{t}  
& = & 
 Y_{0}  +  \int_{0}^{t}   G  \left(s,  \pi_{C}\left( Y_{s} \right) \right) ds
 \\
 &&
+
 \sum_{\ell=1}^{\infty} \int_{0}^{t} 
\left( L_{\ell} \left(s \right) \pi_{C}\left(  Y_{s} \right) 
- \Re \left\langle  Y_{s} , 
L_{\ell}  \left(s \right)  \pi_{C}\left(  Y_{s} \right) \right\rangle  Y_{s} \right)  dW_{s}^{\ell}.
\end{eqnarray*}
 \end{itemize}

We shall say, for short, that 
$\left( \mathbb{Q}, \left( Y_{t}\right)_{t\in\mathbb{I}},\left( W_{t}\right)_{t\in\mathbb{I}}\right)$ 
is a $C$-solution of (\ref{eq:nlSSE}).
\end{definition}

\begin{theorem}
\label{teorema7} 
Let $C$ satisfy Hypothesis \ref{hyp:formal-conservativity}. 
Assume that $\theta$ is a probability measure on $\mathfrak{h}$ concentrated on 
$\mathcal{D}(C)\cap\left\{ y\in\mathfrak{h} :\left\Vert y\right\Vert =1\right\} $ 
such that 
$\int_{\mathfrak{h}}\left\Vert Cx\right\Vert ^{2}\theta\left( dx\right) < \infty$. 
Then (\ref{eq:nlSSE}) has a unique $C$-solution 
$\left(\mathbb{Q},\left( Y_{t}\right) _{t\geq 0},\left( W_{t}\right)_{t\geq0}\right)$ 
with initial law $\theta$.
\end{theorem}

\begin{proof}
Theorem  \ref{teorema2} allows us to use arguments  
of Theorem 1 in  \cite{MoraReAAP} to show our statement. 
\end{proof}

\begin{remark}
 Let the assumptions of Theorem \ref{teorema7} hold,
and let $\left( X_{t}\left( \xi\right) \right) _{t\geq0}$ 
be the strong $C$-solution of (\ref{eq:SSE}), 
where $\xi$ 
is distributed according to $\theta$. 
For a given $T\in\left] 0,+\infty \right[ $,
we define 
$
\mathbb{Q}=\left\Vert X_{T}\left( \xi\right)
\right\Vert ^{2}\cdot\mathbb{P}
$,
$$
B_{t}^{\ell} = W_{t}^{\ell}
-\int_{0}^{t}\frac{1}{\left\Vert X_{s}\left(
\xi\right) \right\Vert ^{2}}d\left[ W^{\ell}, 
\left\| X\left( \xi\right) \right\| ^2 \right]_{s} ,
$$
and 
\[
Y_{t}=\left\{ 
\begin{array}{ll}
X_{t}\left( \xi\right) /\left\Vert X_{t}\left( \xi\right)
\right\Vert , & \;\mathrm{if } \; X_{t} \left( \xi\right) \neq0 \\ 
0, & \;\mathrm{if } \; X_{t}\left( \xi\right) =0%
\end{array}
\right. ,
\]
where $t\in\left[ 0,T\right] $ and $\ell \in\mathbb{N}$. 
By Theorem \ref{teorema2}, 
proceeding along the same lines as in the proof of 
Proposition 1 of  \cite{MoraReAAP}
we can obtain that
$$
\left( \Omega,\mathfrak{F},\left( \mathfrak{F}_{t}\right) _{t\in\left[ 0,T\right] },\mathbb{Q},\left( Y_{t}\right) _{t\in\left[ 0,T\right] },\left( B_{t}^{\ell}\right) _{t\in\left[ 0,T\right] }^{\ell\in \mathbb{N}}\right) 
$$
is a $C$-solution of (\ref{eq:nlSSE}) with initial distribution $\theta$.
\end{remark}

\section{Open quantum systems in coordinate representation} 
\label{subsec:oqs-coord}
We now focus on the model given by  (\ref{eq:Hex}) and (\ref{eq:Lex}),
with the functions $ \sigma_{\ell h}$ satisfying 
\begin{equation}\label{4.2}
\sum_{\ell\ge 1}\sigma_{\ell k}\left(t,x\right)
(\partial_j \overline{\sigma}_{\ell h})(t,x)
=\sum_{\ell\ge 1}\overline{\sigma}_{\ell k}\left(t,x\right)
(\partial_j  \sigma_{\ell h})(t,x)
\end{equation}
for all $j,h,k$. 
It is worth noticing that (\ref{4.2}) obviously holds when functions 
$\sigma_{\ell k}$ do not depend on $x$ and also when they are 
real valued or can be transformed into real valued functions by 
a suitable change of phase. 
A counterexample due to \cite{FaLPM} shows that mean norm square 
conservation may fail when (\ref{4.2}) does not hold and phases of 
$\sigma_{\ell k}$ depend on the space variable $x$.
We next collect our smoothness assumptions on the functions involved in  (\ref{eq:Hex}) and (\ref{eq:Lex}).

\begin{condition}\label{hyp:growth-cond}
Let $L_\ell(t)$ be the operator (\ref{eq:Lex}) 
and assume that (\ref{4.2}) holds. 
For all  $ t \geq 0$,
define 
$
 G \left( t \right) =-iH  \left( t \right)-\frac{1}{2} \sum_{\ell=1}^{m} L_\ell^*\left(t\right) L_\ell\left(t\right)
$,
where $H  \left( t \right)$ is as in (\ref{eq:Hex}).
Suppose that there exists a continuous increasing 
function  $K:[0,+\infty[\to]0,+\infty[$ such that:
\begin{itemize}
 \item[(H4.1)] 
 For all  $t \geq 0$ and $1\le j \le d$,
 $ V  \left(t , \cdot \right) \in C^{2}\left( \mathbb{R}^{d} ,\mathbb{R}\right)$, 
$A^j \left(t , \cdot \right)  \in C^{3}\left( \mathbb{R}^{d} ,\mathbb{R}\right)$.
Moreover,
$
\max \left\{  
\left| V \left( t, x  \right) \right| , \left| {\Delta}V \left( t, x  \right) \right|,  \left|\partial_j (\Delta A^j) \right|\right\}
\leq 
K \left( t \right) \left( 1+ \left| x \right| ^{2} \right)
$,
\[
 \max \left\{  
 \left| \partial_j V \left( t, x  \right) \right|,   \left|A^j\left(t,x\right)\right|, \left|(\partial_{j^\prime}\partial_j A^j)(t,x)\right| \right\}
 \le 
 K \left( t \right) \left( 1+ \left| x \right| \right)
\]
and 
$
 \left|\partial_{j^\prime} A^j\left(t,x\right)\right|   
 \leq 
 K \left( t \right)
$,
where 
$t \geq 0$,  $x \in \mathbb{R}^{d}$ and $1\le j,j^\prime\le d$.

\item[(H4.2)] 
For all $1\le \ell\le  m$ and $t \geq 0$ we have
 $\left|\sigma_{\ell k} \left(t , \cdot \right) \right| \leq K(t)$, with $1 \le k \le d$,
$\eta_{\ell} \left(t , \cdot \right) \in C^3\left( \mathbb{R}^d ,\mathbb{C}\right)$ 
 and 
 the absolute values of all the partial derivatives of  $\eta_{\ell} \left(t , \cdot \right)$
 from the first up to the third order are bounded by $K(t)$.
 Moreover, 
 at least one of the following conditions holds:
 
\subitem{\rm (H4.2.a)}
For all $1\le \ell\le  m$, $1 \le k \le d$ and $t \geq 0$
we have 
 $\left| \eta_{\ell} \left(t , \cdot \right) \right| \leq K(t)$,
 $\sigma_{\ell k}\left(t , \cdot \right) \in C^3\left( \mathbb{R}^d ,\mathbb{C}\right)$,
and 
the absolute values of all partial derivatives of 
$\sigma_{\ell k}\left(t , \cdot \right)$ up to the third order are dominated by $K(t)$.

\subitem{\rm (H4.2.b)}
For any $1\le \ell\le  m$ and $1 \le k \le d$, 
the function $\left( t, x \right) \mapsto \sigma_{\ell k} \left( t, x \right)$ does not depend on $x$
and $\left| \eta_{\ell} \left(t , 0 \right) \right| \leq K(t)$.
\end{itemize}
\end{condition}

Note that condition (H4.2.b) allows linear growth in $x$ of $\eta(t,x)$ while (H4.2.a) does not. 
Theorems \ref{teorema1} and \ref{teorema2} help us to establish the following result.

\begin{theorem}
\label{teo:Smooth} 
Suppose that Hypothesis \ref{hyp:growth-cond} holds and 
set $  C = - {\Delta}+ \left| x \right| ^{2} $.
Let $\xi$ be a $\mathfrak{F}_0$ - measurable random variable 
taking values in 
$L^{2}\left( \mathbb{R}^{d} ,\mathbb{C}\right) $
such that 
$\mathbb{E} \left\Vert  \xi  \right\Vert ^{2} =1$ and 
$\mathbb{E} \left\Vert  C \xi  \right\Vert ^{2}  < \infty$. 
Then (\ref{eq:SSE}) has a unique strong $C$-solution 
with initial datum $\xi$.
Moreover, $\mathbb{E} \left\Vert  X_{t}\left( \xi \right)  \right\Vert ^{2} 
= \left\Vert  \xi   \right\Vert ^{2}$ for all $t > 0$.
If in addition $\left\| \xi \right| = 1$ a.s.,
then
(\ref{eq:nlSSE}) has a unique $C$-solution
whose initial distribution coincides with that of $\xi$.
\end{theorem}

\begin{proof}
 Deferred to Subsection \ref{subsec:teoSmooth}.
\end{proof}

Theorem \ref{teo:Smooth}
applies in a number of physical models
like those listed below,
which, for simplicity,
are restricted to 
$ \mathfrak{h} = L^{2}\left( \mathbb{R} ,\mathbb{C}\right)$
and $m=1$.

\begin{enumerate}
\item[(E.1)] 
Choose 
$\alpha = 1/ \left( 2M \right)$,
 $A^1 \left(t, x \right) = c x$,
$ \sigma_{1 1}(t,x)=b$,
and
$\eta_{1} \left(t , x \right) = a x$,
where $a,b,c \in \mathbb{R}$ and  $M>0$.
Moreover, the potential $V$ is a smooth function.
This describes
a large particle coupled to a bath of harmonic oscillators in thermal equilibrium
(see, e.g.,  \cite{Halliwell}).

\item[(E.2)]\label{exemp:Bassi-Caves}
Let  $ \alpha = 1/\left( 2M \right) $, with $M>0$.
Moreover, we take $A^1 \left(t, x \right) =  \sigma_{1 1}(t,x)=0$ and 
$\eta_{1} \left(t , x \right) = \eta x$,
where  $\eta$ is a real number. 
This model describes the dynamics of the continuous measurement of position of a free 
quantum particle subject to a time-dependent potential $V(t,\cdot)$ 
(see, e.g., \cite{Bassi2009,Gough}),
 a  process that can be observed with detectors.

\item[(E.3)] \label{exemp: Singh-Rost} 
\cite{Singh2007} modeled 
the application of intense linearly polarized laser  to the hydrogen atom
by means of:
 $\alpha = 1/2$, 
$A^1 (t,\cdot) = \sigma_{1 1}(t,\cdot)=0$, $\eta_{1} \left(t , x \right) = -i \eta x$,
and
$$
V \left(t, x \right) = V_0  \left( x \right)   +  x F \left( t \right),
$$
where 
$V_0  \left( x \right) =-1 / \left( x^{2} + \epsilon^{2} \right)^{1/2}$ and 
 \begin{equation*}
F \left( t \right)  = F_{0} \sin \left( \beta t + \delta \right) 
\cdot\left\{ 
\begin{array}{lll}
\sin \left( \pi t / \left( 2 \tau \right) \right), & \text{if } t < \tau \\ 
1, & \text{if } \tau \leq t \leq T - \tau \\ 
\cos^{2} \left( \pi \left( t  + \tau - T \right) /  \left( 2 \tau \right) \right), & \text{if }  T - \tau  \leq t \leq   T 
\end{array}
\right. .
\end{equation*}
Here  $\beta, \eta, \delta \in \mathbb{R}$ and $\epsilon, F_{0}, \tau, T$ are positive constants. 
This  simulates 
the evolution of the electron of the hydrogen atom under the 
influence of a laser field $F \left( t \right)$. 
The soft core potential  $V$  approximates the Coulomb potential of the atom.

\item[(E.4)] \label{exemp:Grotz-Schneider}  
To describe 
the evolution of a quantum system in a parabolic fluctuating trap,
we follow  \cite{Grotz} and \cite{Schneider}
in assuming 
 $ \alpha = 1/\left( 2M \right) $,
$A^{1}(t,x)=\sigma_{1 1}(t,x)=0$, $V \left(t, x \right) =  \frac{1}{2} M \omega^2  x^{2}$
and  $\eta_{1} \left(t , x \right) = -i  \eta x$,
where $M,  \eta > 0$ and $\omega \in \mathbb{R}$.

\item[(E.5)]
A free particle confined by a moving  Gaussian well,
in interaction with a heat bath,
is simulated by 
$\alpha = 1/ \left( 2M \right)$,  $A^1 \left(t, x \right) = 0$,
$$
V \left( t, x \right) = -V_0 \exp  \left(  - \alpha \left( x - r \left(t \right) \right)^2 \right),
$$
$ \sigma_{1 1}(t,x)=b$
and
$\eta_{1} \left(t , x \right) = a x$,
where $a,b \in \mathbb{R}$ and $M,V_0,  \alpha > 0$.
The measurable bounded function $r:\left[ 0 , \infty \right[ \rightarrow \mathbb{R} $  
represents the displacement of the trap's center.
\end{enumerate}


\section{Ehrenfest's theorem}
\label{sec:Ehrenfest}

\subsection{Markovian open quantum systems}

The next theorem provides a  rigorous der\-i\-va\-tion of a version of Ehrenfest's equations 
for open quantum systems in  Lindblad form.

\begin{condition}
\label{hyp:Ehrenfest}
Let $C$ satisfy Hypothesis  \ref{hyp:formal-conservativity}.
Suppose that:
\begin{itemize}
\item[(H5.1)]
For all $t \geq 0$ and any $x$ belonging to a core of $C$,
$$
\sum_{\ell=1}^{\infty }\left\Vert C^{1/2} L_{\ell} \left( t \right) x \right\Vert ^{2}
\leq K \left( t \right)  \left\Vert x\right\Vert_{C}^{2} .
$$
\end{itemize}
Let  $A=B_1^{*} B_2$,
where $B_1, B_2$ are  operators in $\mathfrak{h}$ such that:
\begin{itemize}
\item[(H5.2)]
For all $x \in \mathcal{D}\left( C^{1/2} \right)$,
$
 \max\{\left\Vert  B_1  x  \right\Vert^{2} , \left\Vert  B_2  x  \right\Vert^{2}  \}
\leq K  \left\Vert  x  \right\Vert_{C^{1/2}}^{2}
$.

\item[(H5.3)]
$ \max \left\{   \left\Vert  A  x  \right\Vert^{2}, \left\Vert  A^*  x  \right\Vert^{2} \right\}
 \leq K \left\Vert  x  \right\Vert_{C}^{2}$
whenever $x \in \mathcal{D}\left( C\right)$.

\end{itemize}

 \end{condition}

\begin{theorem}
 \label{th:Ehrenfest}
Let Hypothesis \ref{hyp:Ehrenfest} hold,
together with $\xi\in L_{C}^{2}\left(\mathbb{P};\mathfrak{h}\right)$.
Then, for all $t \geq 0$ we have
\begin{eqnarray}
\label{eq:10}
 \mathbb{E} \left\langle X_{t} \left( \xi \right) , A X_{t} \left( \xi \right) \right\rangle
 & = & 
\mathbb{E} \left\langle  \xi  , A  \xi  \right\rangle
+
\int_0^t 
\mathbb{E} \left\langle A^* X_{s} \left( \xi \right) , G \left( s \right) X_{s} \left( \xi \right) \right\rangle
 ds
\\
\nonumber
&&
+ \int_0^t 
\mathbb{E} \left\langle G \left( s \right) X_{s} \left( \xi \right) , A  X_{s} \left( \xi \right) \right\rangle 
ds
\\
\nonumber
&&
+
\int_0^t \left(
\sum_{\ell = 1}^{\infty}
\mathbb{E} \left\langle B_1 L_{\ell} \left( s \right) X_{s} \left( \xi \right) , B_2 L_{\ell} \left( s \right) X_{s} \left( \xi \right) \right\rangle
\right) ds .
\end{eqnarray}
\end{theorem}

\begin{proof}
 Deferred to Subsection \ref{subsec:teoEhrenfest}.
\end{proof}

Suppose that $X_{t} \left( \xi \right)$ is the unique strong $C$-solution of (\ref{eq:SSE}).
Set
$$
 \rho_t := \mathbb{E} \left|X_{t} \left( \xi \right)\right\rangle \left\langle X_{t} \left( \xi \right)\right| ,
$$
where we use Dirac notation.
Then $\rho_t $ is a $C$-regular density operator 
and 
$$
tr \left( \rho_t A \right) =  \mathbb{E} \left\langle X_{t} \left( \xi \right) , A X_{t} \left( \xi \right) \right\rangle ,
$$
provided that  $A$ is $C$-bounded
(see  \cite{MoraAP} for details).
In the homogeneous case,
from  \cite{MoraAP} we have that 
$ \rho_t $ is the unique solution of the quantum master equation
\begin{equation*}
 \dfrac{d}{dt} \rho_{t}
 = 
 G \rho_t +  \rho_t  G^{\ast} +\sum_{\ell=1}^{\infty}L_{\ell}  \rho_t  L_{\ell}^{\ast},
 \hspace{1cm}
 \rho_{0} = \mathbb{E} \left| \xi \right\rangle \left\langle \xi \right| .
\end{equation*}
We now combine Theorem \ref{th:Ehrenfest} with Theorem 3.2 of \cite{MoraAP}
to deduce the following corollary,
which asserts which asserts the validity (\ref{eq:3}) whenever essentially $A L_{\ell}$ is $C$-bounded.
To this end,
we use basic properties of the adjoints of unbounded operators (see, e.g., \cite{Kato}).

\begin{corollary}
\label{cor:Ehrenfest}
In addition to Hypothesis \ref{hyp:Ehrenfest} and $\xi\in L_{C}^{2}\left(\mathbb{P};\mathfrak{h}\right)$,
suppose that the operators 
$G \left( t \right), B_1 L_1  \left( t \right), B_2 L_2  \left( t \right), \ldots $ are cerrable for all $t \geq 0$. 
Then 
\begin{eqnarray}
\label{eq:14}
tr \left(  A \rho_t \right)
& = & 
tr \left(  A \rho_0 \right)
+
\int_0^t 
\left(
tr \left(G \left( s \right)  \rho_s A \right)
+
tr \left(A \rho_s  G \left( s \right)^{*} \right)
\right) ds
\\
\nonumber
&&
+
\int_0^t \left(
\sum_{\ell = 1}^{\infty}
tr \left( B_2 L_{\ell} \left( s \right) \rho_s L_{\ell} \left( s \right)^{*} B_1 ^{*} \right)
\right) ds,
\end{eqnarray}
where 
$t \geq 0$
and
$ \rho_t := \mathbb{E} \left|X_{t} \left( \xi \right)\right\rangle \left\langle X_{t} \left( \xi \right)\right| $.
\end{corollary}

\subsection{Applications}

We begin by applying Theorem \ref{th:Ehrenfest} to the model given by (\ref{eq:Hex}) and (\ref{eq:Lex}).

\begin{theorem}
  \label{th:Ehrenfest_Model}
 Assume the context of (\ref{eq:Hex}) and (\ref{eq:Lex}),
 together with Hypothesis \ref{hyp:growth-cond}.
 Let $A=B_{1}^{*} B_{2}$,
 where $B_{1}$ and $B_{2}$ satisfy one of the following conditions:
\begin{itemize}
\item
$B_1 = \mlt{c_1}$ and $B_2 = \mlt{c_2}$
provided that  $c_1, c_2: \mathbb{R}^{d} \rightarrow \mathbb{R}$ are  Borel measurable functions
such that
$
 \left|   c_j \left( x  \right)  \right|
 \leq
 K\left( 1+ \left| x \right| \right)
$
for all $x \in  \mathbb{R}^{d}$
and $j=1,2$.

\item
For  any $j=1,2$,
$B_j$ is either $\partial_k \mlt{a_j}$, $\mlt{b_j} \partial_k$ or $\mlt{c_j}$,
 where  
 $k=1,\ldots,d$, 
 $a_j  \in C^{2}\left( \mathbb{R}^{d} ,\mathbb{R}\right) $
 and
 $b_j, c_j \in C^{1}\left( \mathbb{R}^{d} ,\mathbb{R}\right)$.
 Moreover, 
 for all $x \in \mathbb{R}^d$ and $l,k = 1, \ldots, d$ we have:
 $
 \max \left\{  \left| a_j \left( x  \right) \right| , \left| b_j \left( x  \right) \right|  \right\} \leq K 
 $,
 $$
 \max \left\{  
 \left| c_j \left( x  \right)  \right| ,  \left| \partial_l a_j \left( x  \right) \right| ,\left|  \partial_l b_j \left( x  \right) \right|
\right\}
\leq 
K\left( 1+ \left| x \right| \right),
 $$
 and
$
\max \left\{  
 \left|  \partial_l  c_j \left( x  \right)  \right| ,  \left| \partial_k \partial_l a_j \left( x  \right) \right|
\right\}
\leq 
K \left( 1+ \left| x \right|^2 \right) 
$.
\end{itemize}
 Then
(\ref{eq:10}) and (\ref{eq:14}) hold in case
$\xi\in L_{- {\Delta}+ \left| x \right| ^{2}}^{2}\left(\mathbb{P};\mathfrak{h}\right)$.
\end{theorem}
\begin{proof}
 Deferred to Subsection \ref{subsec:teoEhrenfest_Model}.
\end{proof}

Using Theorem \ref{th:Ehrenfest_Model} we can obtain  expressions 
describing the evolution of some important observables,
which sometimes are closed systems of ordinary differential equations.
For instance, 
the following theorem makes mathematically rigorous 
computations given in \cite{Schneider},
which establish the linear heating of a Paul trap 
due to fluctuating electrical fields that change the center of this ion trap
(see also \cite{Gehm1998,Grotz}).

\begin{corollary}
  \label{cor:heat}
Consider (\ref{eq:Hex}) and (\ref{eq:Lex})  with 
$d=1$, 
$ \alpha = 1 / (2M) $, 
$A^{j}(t,x)=0$, 
$V \left(t, x \right) = V \left( x \right) $,
$\sigma_{\ell k}(t,x)=0$
and  $\eta_{1} \left(t , x \right) = -i  \eta x$,
where 
$M, \eta > 0$
and 
$V \in  C^{2}\left( \mathbb{R} ,\mathbb{R}\right) $.
Suppose that for any $x \in \mathbb{R}$,
$
 \left|  V \left( x  \right) \right|
 \le 
 K \left( 1+ \left| x \right|^2 \right)
$,
$
 \left|  V^{\prime} \left( x  \right) \right| 
 \le 
 K \left( 1+ \left| x \right| \right)
$
and 
$
 \left|  V^{\prime \prime} \left( x  \right) \right|
 \le 
 K \left( 1+ \left| x \right|^2 \right)
$.
Then for all $t \geq 0$,
\begin{equation}
\label{eq:17}
 \mathbb{E} \left\langle X_{t} \left( \xi \right) , H X_{t} \left( \xi \right) \right\rangle
=
\mathbb{E} \left\langle  \xi  , H \xi  \right\rangle
+
\frac{1}{2M} \eta^2 t .
\end{equation}
\end{corollary}

\begin{proof}
 Deferred to Subsection \ref{subsec:cor_heat}.
\end{proof}

\begin{remark}
 \cite{Schneider}  
 restricted their attention to 
 $$V \left( x \right) =  M \omega^2  x^{2} /2. $$
\end{remark}

\section{Proofs}
\label{sect:proofs}

\subsection{Proof of  Lemma \ref{lema:Measurability:g}}
\label{subsec:lemaMeasurability:g}

We first characterize the domain of $C$ by means of Yosida approximations of $-C$.

\begin{lemma}
\label{lema22}
Let $C$ be a self-adjoint positive operator in $\mathfrak{h}$.
Then 
\begin{equation*}
 \mathcal{D}\left( C \right)  
 =  
 \left\{ x \in \mathfrak{h}: \left( C R_{n}  x \right)_{n }  \medspace  converges \right\} 
 = 
 \left\{ x \in \mathfrak{h}: \sup_{n \in \mathbb{N}}  \left\Vert  C R_{n}  x \right\Vert < \infty  \right\},
\end{equation*}
where $R_{n} = n  \left(n+C \right)^{-1}$.
\end{lemma}

\begin{proof}
 Since $-C$ is dissipative and self-adjoint, 
 for all $x \in  \mathcal{D}\left( C\right)$ we have
$$
C R_{n} x \longrightarrow_{n \rightarrow \infty} Cx
$$
 (see, e.g.,  Pazy \cite{Pazy}). Thus 
$ \mathcal{D}\left( C \right) \subset \left\{ x \in \mathfrak{h}: \left( C R_{n}  x \right)_{n }  \medspace  converges \right\} 
$.

Now, assume that 
$\left( \left\Vert  C  R_{n} x \right\Vert \right)_{n \in \mathbb{N}}$ is bounded. 
Using the Banach-Alaoglu theorem we deduce that there 
exists a subsequence 
$\left( C  R_{n_{k}} x \right)_{k \in  \mathbb{N}}$ 
which converges weakly to a vector $z \in \mathfrak{h}$. 
Since $R_{n} x \longrightarrow_{n \rightarrow \infty} x$,
for any $y \in  \mathcal{D}\left( C\right)$ we have
\[
 \left\langle  x, Cy \right\rangle
  =    \lim_{k \rightarrow \infty}  \left\langle  R_{n_{k}} x, C y \right\rangle 
  =    \lim_{k \rightarrow \infty}  \left\langle  CR_{n_{k}} x, y \right\rangle 
  =   \left\langle  z, y \right\rangle.
\]
Hence $x \in \mathcal{D}\left(C^{\ast} \right)$ ($= \mathcal{D}\left(C \right)$), and so 
$
 \left\{ x \in \mathfrak{h}: \sup_{n \in \mathbb{N}}  \left\Vert   C R_{n} x \right\Vert < \infty  \right\} \subset \mathcal{D}\left( C \right)
$.
\end{proof}

The assertion of Lemma \ref{lema:Measurability:g} follows straightforward from the next lemma.

\begin{lemma}
 \label{lema:Measurability}
Let $C$ be a self-adjoint positive operator on $\mathfrak{h}$.
Suppose that
$L \in \mathfrak{L}\left( \left( \mathcal{D}\left( C\right), \left\Vert \cdot \right\Vert _{C}\right)  ,\mathfrak{h}\right)$.
Then 
$L \circ\pi_{C}: \left(\mathfrak{h}, \mathcal{B} \left( \mathfrak{h} \right) \right) \rightarrow \left(\mathfrak{h}, \mathcal{B} \left( \mathfrak{h} \right) \right)$
 is measurable.
\end{lemma}

\begin{proof}
Let $R_{n}$ be as in Lemma \ref{lema22}.
Using Lemma \ref{lema22} we obtain that  $\mathcal{D}\left( C\right)$ is a Borel set of $\mathfrak{h}$
since $C R_{n} \in \mathfrak{L}\left( \mathfrak{h} \right)$,
and so
$\pi_{C}: \left(\mathfrak{h}, \mathcal{B} \left( \mathfrak{h} \right) \right) \rightarrow \left(\mathfrak{h}, \mathcal{B} \left( \mathfrak{h} \right) \right)$ 
is measurable.
Since  the range of $R_{n}$ is a subset of  $\mathcal{D}\left( C\right)$
and  $L \in \mathfrak{L}\left( \left( \mathcal{D}\left( C\right), \left\Vert \cdot \right\Vert _{C}\right)  ,\mathfrak{h}\right)$,
$LR_{n} \in  \mathfrak{L}\left( \mathfrak{h} \right)$.
 Hence $L  R_{n} \circ \pi_{C}$ is measurable.
 It follows from 
 $R_{n}  \longrightarrow_{n \rightarrow \infty} I$ and 
 $$ C R_{n} x \longrightarrow_{n \rightarrow \infty} Cx$$
that
$
L R_{n}  \circ \pi_{C} \longrightarrow_{n \rightarrow \infty}  
L \circ \pi_{C}
$,
which implies the measurability of $ L \circ \pi_{C}$.
\end{proof}

\subsection{Proof of Theorem \ref{teorema1}}
\label{subsec:existence}

First, we extend  the inequality given in Condition  H2.3  to  $\mathcal{D}\left( C^{2} \right)$.

\begin{remark}
 \label{nota2}
 Let $L$ be a closable operator in $\mathfrak{h}$ such that 
$\mathcal{D}\left( C\right) \subset \mathcal{D}\left( L\right) $, 
with $C$  self-adjoint positive operator in $\mathfrak{h}$.  
Applying the closed graph theorem gives 
$L \in \mathfrak{L}\left( \left( \mathcal{D}\left( C\right) ,\left\Vert \cdot\right\Vert _{C}\right) ,\mathfrak{h}\right)$. 
\end{remark}

\begin{lemma}
\label{lema2i}
Suppose that $C$ satisfies Conditions H2.1 - H2.3 of Hypothesis \ref{hyp:CF-inequality}.
If $x$ belongs to $\mathcal{D}\left( C^{2} \right)$ and $t \geq 0$,
then 
$L_{\ell}\left( t \right) x \in  \mathcal{D}\left( C \right) $
for any $\ell \in \mathbb{N}$, and 
\begin{equation}
 \label{eq:2}
  2\Re\left\langle C^{2} x, G \left( t \right) x\right\rangle 
 +\sum_{\ell=1}^{\infty }\left\Vert C L_{\ell} \left( t \right) x \right\Vert ^{2}
 \leq \alpha \left( t \right) \left\Vert x\right\Vert_{C}^{2} .
\end{equation}
\end{lemma}

\begin{proof}
Since $\mathfrak{D}_{1}$ is a core of $C^{2}$,  there exists a sequence $\left( x_{n} \right)_{n \in \mathbb{N}}$ in $\mathfrak{D}_{1}$ converging to $x$ such that
$
C^{2}x_{n} \longrightarrow _{n \rightarrow \infty} C^{2}x
$.
Using Remark \ref{nota2} and Condition H2.1 we deduce that
 $G \left( t \right) \in \mathfrak{L}\left( \left( \mathcal{D}\left( C^{2}\right), \left\Vert \cdot \right\Vert _{C^{2}}\right),\mathfrak{h}\right)$,
 and so   Condition H2.3 leads to
\begin{equation}
 \label{eq:1}
 \lim _{n, n^\prime \rightarrow \infty} \sum_{\ell=1}^{\infty }
\left\Vert C L_{\ell} \left( t \right) \left(x_{n}-x_{n^\prime}\right) \right\Vert ^{2} =0.
\end{equation}
By $C$ is closed,  
from (\ref{eq:1}) we have 
$ L_{\ell} \left( t \right) x \in \mathcal{D}\left( C \right)$ 
and
$C L_{\ell} \left( t \right) x_{n} \rightarrow C L_{\ell} \left( t \right) x$ as $n \rightarrow \infty$.
Then (\ref{eq:2}) follows immediately, 
because (\ref{eq:2}) is true for $x_n$ for all $n$. 
 
\end{proof}

The inequality of Condition H2.4 can be immediately extended to $\mathcal{D}\left( C \right)$, by the definition of core and Fatou's lemma, following the lines of the proof of Lemma \ref{lema2i}.

\begin{lemma}
\label{lema5}
Under Conditions H2.1, H2.2 and H2.4, 
for all $x$ in $\mathcal{D}\left( C \right)$ we have
\[
2\Re\left\langle  x, G \left( t \right) x\right\rangle 
+\sum_{k=1}^{\infty }\left\Vert  L_{\ell} \left( t \right) x \right\Vert ^{2}\leq 0.
\]
\end{lemma}

In contrast to  \cite{MoraReIDAQP}, where we used the Galerkin method, 
in the proof of Theorem \ref{teorema1} we obtain $X_{t}\left( \xi \right)$
as the $L^{2}\left( \mathbb{P} , \mathfrak{h}\right)$-weak limit of the 
solutions to the sequence of stochastic evolution equations (\ref{3.29}) given below.

\begin{definition}
\label{definicion1}
Let Hypothesis \ref{hyp:L-G-C-domain} hold,
together with Conditions H2.1 and H2.2.
Suppose that $\xi$ is a $\mathfrak{F}_{0}$-measurable random variable belonging 
to $L^{2}\left( \mathbb{P},\mathfrak{h}\right) $. For each natural number $n$, 
 we  define $X^{n}$   to be the unique continuous solution of 
 \begin{equation}
\label{3.29}
X_{t}^{n}  =
\xi +\int_{0}^{t}G^{n}\left( s \right) X_{s}^{n}  ds 
+ \sum_{\ell=1}^{n}\int_{0}^{t}L_{\ell}^{n} \left( s \right) X_{s}^{n} dW_{s}^{\ell},
\end{equation}
where  $G^{n} \left( s \right) = \widetilde{R}_{n}G \left( s \right) \widetilde{R}_{n}$ 
and $L_{\ell}^{n}  \left( s \right) = L_{\ell} \left( s \right) \widetilde{R}_{n}$ with 
$
\widetilde{R}_{n} = n  \left(n+C^{2} \right)^{-1}
$.
\end{definition}

\begin{remark}
Recall that  $C^{2} \widetilde{R}_{n} \in \mathfrak{L}\left( \mathfrak{h} \right)$
and $\left\Vert \widetilde{R}_{n} \right\Vert \leq 1$. As a consequence, $X^{n}$ is well-defined
because H2.1 and H2.2 imply that  $G^{n} \left( t \right)$ and $L_{\ell}^{n}  \left( t \right)$
are  bounded operators in $ \mathfrak{h}$
whose norms are uniformly bounded on compact time intervals.
\end{remark}

Though the next three   estimates for $ X^{n} $ 
essentially coincide with those given in Lemma 2.3 of  \cite{MoraReIDAQP},
the infinite-dimensional nature of (\ref{3.29})
forces us to 
use a more refined analysis.

\begin{lemma}
 \label{lema1}
 Adopt Hypothesis \ref{hyp:L-G-C-domain}, together with Conditions H2.1, H2.2 and H2.4.
Then for any $t \geq 0$,
$
 \mathbb{E} \left\Vert X_{t}^{n} \right\Vert ^{2} \leq  \mathbb{E} \left\Vert \xi \right\Vert ^{2} 
 $.
 Moreover, for all $x \in \mathfrak{h}$ and  $t \geq 0$ we have
 \begin{equation}
\label{3.37}
2\Re\left\langle  x, G^{n} \left( t \right) x\right\rangle 
+\sum_{\ell=1}^{\infty }\left\Vert  L_{\ell}^{n} \left( t \right) x \right\Vert ^{2}
\leq 0.
\end{equation}
\end{lemma}

\begin{proof}
Since the range of $\widetilde{R}_{n}$ is a subset of $\mathcal{D}\left( C^{2} \right)$,
Lemma \ref{lema5} leads to (\ref{3.37}). Using complex It\^o's formula we obtain
\begin{equation}
\label{3.30}
\left\Vert X_{t}^{n} \right\Vert ^{2} \leq  \left\Vert \xi \right\Vert ^{2} 
+ \sum_{\ell=1}^{n} \int_{0}^{t} 2\Re \left\langle  X_{s}^{n}, L_{\ell}^{n} 
\left( s \right) X_{s}^{n} \right\rangle dW_{s}^{\ell}.
\end{equation}
Set $\tau _{j} = \inf{ \left\{  t \geq0: \left\Vert  X_{t}^{n}  \right\Vert > j \right\} }$. 
Then $\tau _{j} \nearrow \infty$ as $j \rightarrow \infty$,
because $X^{n}$ is pathwise continuous.
By (\ref{3.30}), Fatou's lemma  yields
$
\mathbb{E} \left\Vert X_{t}^{n} \right\Vert ^{2}  \leq  
\liminf_{j \rightarrow \infty} \mathbb{E} \left\Vert X_{t \wedge \tau _{j} }^{n} \right\Vert ^{2} 
   \leq  \mathbb{E} \left\Vert \xi \right\Vert ^{2}
$.
\end{proof}

\begin{lemma}
 \label{lema3}
 Let Hypothesis \ref{hyp:CF-inequality} hold.
 If $\xi \in L_{C}^{2}\left( \mathbb{P},\mathfrak{h}\right) $,
 then
 \begin{equation}
 \label{3.35}
\mathbb{E} \left\Vert C  X_{t}^{n} \right\Vert ^{2} 
\leq 
\exp \left( t \alpha \left( t \right) \right) \left( \mathbb{E} \left\Vert C  \xi  \right\Vert ^{2} 
+ t \alpha\left( t \right)  \mathbb{E} \left\Vert  \xi  \right\Vert ^{2} \right).
\end{equation}
\end{lemma}

\begin{proof}
 
 Combining Condition H2.1 with  Lemma \ref{lema2i} we obtain that
 $CG^{n} \left( t \right) $ and $ C L_{\ell}^{n} \left( t \right)$ 
 are bounded operators on $\mathfrak{h}$ 
 whose norms are uniformly bounded on compact intervals.
 Lemma  \ref{lema1}  gives
 $\mathbb{E} \left\Vert C G^{n} \left( t \right) X_{t}^{n} \right\Vert ^{2}
  \leq 
  \left\Vert C G^{n} \left( t \right) \right\Vert ^{2} \mathbb{E} \left\Vert \xi \right\Vert ^{2}
 $
 and 
 $$
 \mathbb{E} \left\Vert C L_{\ell}^{n} \left( t \right) X_{t}^{n} \right\Vert ^{2} 
 \leq  
 \left\Vert C L_{\ell}^{n} \left( t \right) \right\Vert ^{2} \mathbb{E} \left\Vert \xi \right\Vert ^{2} .
$$
Therefore
$
CX_{t}^{n} = Y_{t}^{n}
$
$a.s.$ for any $t \geq 0$,
where
\[
Y^{n} = C \xi +\int_{0}^{\cdot}CG^{n} \left( s \right) X_{s}^{n} ds 
+ \sum_{\ell=1}^{n}\int_{0}^{\cdot}CL_{\ell}^{n} \left( s \right) X_{s}^{n} dW_{s}^{\ell}.
\]
This follows from, for instance, Propositions 1.6 and 4.15 of  \cite{DaPrato}.

Since $\widetilde{R}_{n}$ commutes with both $C$ and $C^2$,
using Lemma \ref{lema2i} and $\left\Vert \widetilde{R}_{n} \right\Vert \leq 1$
we deduce that for any  $x \in \mathcal{D}\left( C^{2} \right)$ and $t \geq 0$,
 \begin{eqnarray*}
&& 
2\Re\left\langle C x, CG^{n} \left( t \right) x\right\rangle 
 +\sum_{\ell=1}^{n}\left\Vert C L_{\ell}^{n} \left( t \right) x \right\Vert ^{2} 
 \\
 & \leq &
 2\Re \left\langle C^{2}\widetilde{R}_{n} x, G \left( t \right) \widetilde{R}_{n}x \right\rangle
 +\sum_{\ell=1}^{ \infty}\left\Vert C L_{\ell} \left( t \right) \widetilde{R}_{n}x \right\Vert ^{2} 
 \leq 
  \alpha \left( t \right)  \left\Vert \widetilde{R}_{n}x\right\Vert_{C}^{2} 
 \leq 
  \alpha \left( t \right)  \left\Vert x\right\Vert_{C}^{2}.
 \end{eqnarray*}
 As $\mathcal{D}\left( C^{2} \right)$ is a core of $C$, 
by a passage to the limit we get that for all $y \in \mathcal{D}\left( C \right)$ and $t \geq 0$,
\begin{equation}
\label{3.34}
 2\Re\left\langle C y, CG^{n} \left( t \right) y\right\rangle 
 +\sum_{\ell=1}^{n}\left\Vert C L_{\ell}^{n} \left( t \right) y \right\Vert ^{2} 
 \leq 
 \alpha \left( t \right) \left\Vert y\right\Vert_{C}^{2} .
 \end{equation}

Finally, 
choose  $\tau _{j} = \inf{ \left\{  t \geq0: \left\Vert  Y_{t}^{n}  \right\Vert > j \right\} }$.
Applying It\^o's formula yields
\[
 \mathbb{E}  \left\Vert Y_{t \wedge \tau _{j} }^{n}  \right\Vert ^{2}   = 
 \mathbb{E}  \left\Vert C \xi  \right\Vert ^{2} 
 +  \mathbb{E} \int_{0}^{t \wedge \tau _{j} } \left(
 2\Re\left\langle Y_{s}^{n}, CG^{n}  \left( s \right)  X_{s}^{n}\right\rangle 
 +\sum_{\ell=1}^{n}\left\Vert C L_{\ell}^{n}  \left( s \right)  X_{s}^{n} \right\Vert ^{2} \right) ds,
\]
because 
$
\mathbb{E} \left\vert \Re \left\langle  Y_{s \wedge \tau _{j} }^{n}, C L_{\ell}^{n} \left( s \right) X_{s}^{n}  \right\rangle  \right\vert ^{2} 
\leq 
j^{2} \left\Vert C L_{\ell}^{n}  \left( s \right)  \right\Vert ^{2} \mathbb{E} \left\Vert \xi \right\Vert ^{2}
$
by  Lemma \ref{lema1}.
Since $Y_{s}^{n} = CX_{s}^{n}$ a.s.,
combining  (\ref{3.34}) with Lemma \ref{lema1} we have
\[
 \mathbb{E}  \left\Vert Y_{t \wedge \tau _{j} }^{n}  \right\Vert ^{2} \leq 
  \mathbb{E}  \left\Vert C \xi \right\Vert ^{2} 
+  \alpha \left( t \right) \int_{0}^{t} \mathbb{E}  \left\Vert C  X_{s}^{n}  \right\Vert ^{2}  ds
 +  t  \alpha \left( t \right) \mathbb{E}  \left\Vert \xi \right\Vert ^{2}   ,
\]
and so
 $$
 \mathbb{E}  \left\Vert Y_{t}^{n}  \right\Vert ^{2} 
  \leq   \liminf_{j \rightarrow \infty}
 \mathbb{E}  \left\Vert Y_{t \wedge \tau _{j} }^{n}  \right\Vert ^{2} 
   \leq 
    \mathbb{E}  \left\Vert C \xi \right\Vert ^{2} 
  +  t \alpha \left( t \right)  \mathbb{E}  \left\Vert \xi \right\Vert ^{2} 
  +  \alpha \left( t \right) \int_{0}^{t} \mathbb{E}  \left\Vert Y_{s}^{n}  \right\Vert ^{2}  ds .
$$
The Gronwall-Bellman lemma now leads to (\ref{3.35}).
\end{proof}

\begin{lemma}
 \label{lema4}
 Fix $T >0$.
 Under the assumptions of Theorem \ref{teorema1},
 \begin{equation}
 \label{3.36}
\mathbb{E} \left\Vert  X_{t}^{n} - X_{s}^{n} \right\Vert ^{2} 
\leq K_{T, \xi} \left( t - s \right),
\end{equation}
where $0 \leq s \leq t < T$ and $K_{T, \xi}$ is a constant depending of $T$ and $\xi$.
\end{lemma}

\begin{proof}
Consider $\tau _{j} = \inf{ \left\{  t \geq0: \left\Vert  X_{t}^{n}  \right\Vert > j \right\} }$. 
According to  It\^o's formula we have
\begin{align*}
 & 
 \mathbb{E} \left\Vert X_{t  \wedge \tau _{j}}^{n} - X_{s  \wedge \tau _{j}}^{n} \right\Vert ^{2}
 \\
 & 
 = 
\mathbb{E} \int_{s  \wedge \tau _{j}}^{t  \wedge \tau _{j}}
\left( 
2\Re\left\langle X_{r}^{n} - X_{s  \wedge \tau _{j}}^{n}, G^{n} \left( r \right) X_{r}^{n}\right\rangle 
+ \sum_{\ell=1}^{n}\left\Vert L_{\ell}^{n} \left( r \right) X_{r}^{n} \right\Vert ^{2} 
\right) ds,
\end{align*}
 and hence  (\ref{3.37}) leads to
\[
\mathbb{E} \left\Vert X_{t  \wedge \tau _{j}}^{n} - X_{s  \wedge \tau _{j}}^{n} \right\Vert ^{2} 
 \leq
- \mathbb{E} \int_{s  \wedge \tau _{j}}^{t  \wedge \tau _{j}} 2\Re\left\langle X_{s}^{n}, G^{n} \left( r \right) X_{r}^{n}\right\rangle dr.
\]
From Condition H2.1, 
$\left\Vert \widetilde{R}_{n} \right\Vert \leq 1$ and $\widetilde{R}_{n}C \subset C\widetilde{R}_{n}$ 
we deduce that 
$
\left\Vert G^{n} \left( t \right) x \right\Vert ^{2} \leq K \left( t \right) \left\Vert x \right\Vert_{C} ^{2}
$
for all $x \in \mathcal{D}\left( C \right)$.
Therefore
$$
\mathbb{E} \left\Vert X_{t  \wedge \tau _{j}}^{n} - X_{s  \wedge \tau _{j}}^{n} \right\Vert ^{2} 
 \leq
K \left( t \right) \mathbb{E} \int_{s  \wedge \tau _{j}}^{t  \wedge \tau _{j}} \left\Vert  X_{s}^{n} \right\Vert \left\Vert X_{r}^{n} \right\Vert_{C} dr
$$
by $X_{s}^{n} \in \mathcal{D}\left( C \right)$ a.s.,
and so Fatou's lemma  implies 
\[
\mathbb{E} \left\Vert X_{t}^{n} - X_{s}^{n} \right\Vert ^{2}  
 \leq 
\liminf_{j \rightarrow \infty} \mathbb{E} \left\Vert X_{t  \wedge \tau _{j}}^{n} - X_{s  \wedge \tau _{j}}^{n} \right\Vert ^{2}  
  \leq  
 K \left( t \right) \int_{s}^{t} \sqrt{\mathbb{E} \left\Vert  X_{r}^{n}  \right\Vert_{C} ^{2} } \sqrt{\mathbb{E} \left\Vert X_{s}^{n}  \right\Vert ^{2} } 
 dr.
\]
Applying Lemmata \ref{lema1} and \ref{lema3} we obtain (\ref{3.36}).
\end{proof}

We next obtain a strong $C$-solution of (\ref{eq:SSE}) by means of a limit procedure.

\begin{definition}
 For any natural number $n$, we define $\left(\mathfrak{G}_{s}^{\xi,n}\right)  _{s\geq0}$ to be the filtration that satisfies the usual hypotheses generated by $\xi$ and $W^{1},\ldots,W^{n}$. Let $t$ be a non-negative real number. By  $\mathfrak{G}_{t}^{\xi,W}$  we mean 
 the $\sigma$-algebra  generated by $\cup_{n\in\mathbb{N}}\mathfrak{G}_{t}^{\xi,n}$. As usual,  $\mathfrak{G}_{t+}^{\xi,W} = \cap_{\epsilon >0}\mathfrak{G}_{t+\epsilon}^{\xi,W} $.
 \end{definition}

 \begin{lemma}
\label{lema30}
Let the assumptions of Theorem \ref{teorema1} hold.
Fix $T > 0$. 
Then, we can extract from any subsequence of  
$\left( X^{n}   \right)_{n \in \mathbb{N}}$ a subsequence $\left( X^{n_{k}}   \right)_{k \in \mathbb{N}}$ for which there exists a  $\left(\mathfrak{G}_{t+}^{\xi,W}\right)_{t\in\left[  0,T\right]}$-predictable process $\left( Z_{t} \right) _{t\in\left[  0,T\right]}$ such that for any $t\in\left[  0,T\right]$,
\begin{equation}
\label{3.39}
X_{t}^{n_{k}}  \longrightarrow_{k\rightarrow\infty} Z_{t}  \qquad weakly\ in\ L^{2}\left(
\left( \Omega,\mathfrak{G}_{t}^{\xi ,W},\mathbb{P} \right), \mathfrak{h}\right).
\end{equation}
\end{lemma}

\begin{proof}
By Lemmata \ref{lema1} and \ref{lema4},
using a compactness method
in the same way as in the proof of Lemma 2.4  of   \cite{MoraReIDAQP}
we obtain our assertion (see Subsection \ref{subsec:lemma30} for details).
\end{proof}

In contrast with  \cite{MoraReIDAQP},
in the following steps 
we do not  make assumptions about
the adjoints of 
$G \left( t \right)$ and $L_{\ell} \left( t \right) $,
which even may not exist.

\begin{lemma}
\label{lema31}
Adopt the assumptions of Theorem \ref{teorema1},
together with the notation of Lemma \ref{lema30}.
Let $t \in \left[  0,T\right]$.
Then 
$
\mathbb{E}\left\Vert Z_{t} \right\Vert ^{2}\leq \mathbb{E}\left\Vert \xi\right\Vert ^{2}
$,
 $Z_{t} \in Dom(C)$ a.s., and 
\begin{equation}
\label{3.42}
\mathbb{E}\left\Vert CZ_{t}  \right\Vert ^{2}
\leq
 \exp \left( \alpha \left( t \right) t \right) \left(  \mathbb{E}\left\Vert C \xi \right\Vert ^{2}
+\alpha  \left( t \right) t   \mathbb{E}\left\Vert \xi\right\Vert ^{2} \right)  .
\end{equation}
Moreover,
$
G^{n_{k}} \left( t \right) X_{t}^{n_{k}}
 \longrightarrow_{k\rightarrow\infty} 
 G \left( t \right) Z_{t}
 $
 weakly in 
 $
 L^{2}\left( \left( \Omega,\mathfrak{G}_{t}^{\xi ,W},\mathbb{P} \right), \mathfrak{h}\right)
 $,
 and for all 
 $ \ell \in \mathbb{N}$, 
\begin{equation}
\label{3.43}
L_{\ell}^{n_{k}} \left( t \right) X_{t}^{n_{k}}   \longrightarrow_{k\rightarrow\infty} 
L_{\ell} \left( t \right) Z_{t} 
\quad weakly\ in\ L^{2}\left(\left( \Omega,\mathfrak{G}_{t}^{\xi ,W},\mathbb{P} \right), \mathfrak{h}\right).
\end{equation}
\end{lemma}

\begin{proof}
By  (\ref{3.39}),
Lemma \ref{lema1} leads to
 $\mathbb{E}\left\Vert Z_{t} \right\Vert ^{2}\leq \mathbb{E}\left\Vert \xi\right\Vert ^{2}$.
 
 For a given  
 $U \in L^{2}\left(\left( \Omega,\mathfrak{G}_{t}^{\xi ,W},\mathbb{P} \right), \mathfrak{h}\right)$,
 the dominated convergence theorem yields
 $\widetilde{R}_{n} U  \longrightarrow_{n \rightarrow \infty} U$ in $L^{2}\left(\left( \Omega,\mathfrak{G}_{t}^{\xi ,W},\mathbb{P} \right), \mathfrak{h}\right)$,
 and so using (\ref{3.39})  we get 
 \begin{equation}x
\label{3.41}
\mathbb{E} \left\langle U, \widetilde{R}_{n_{k}} X^{n_{k}}_{t} \right\rangle 
=
\mathbb{E} \left\langle \widetilde{R}_{n_{k}}U, X^{n_{k}}_{t} \right\rangle
 \longrightarrow_{k\rightarrow\infty} 
 \mathbb{E} \left\langle U, Z_{t}  \right\rangle.
\end{equation}

Suppose that 
$L \in \mathfrak{L}\left( \left( \mathcal{D}\left( C\right), \left\Vert \cdot \right\Vert _{C}\right)  ,\mathfrak{h}\right)$,
and define  $L^{n}=L \widetilde{R}_{n}$.
 Since $\widetilde{R}_{n}C \subset C \widetilde{R}_{n}$ and 
 $\left\Vert \widetilde{R}_{n} \right\Vert \leq 1$, 
 applying Lemma \ref{lema3} and the Banach-Alaoglu theorem we deduce that any subsequence of 
 $\left( n_{k}  \right)_{k \in \mathbb{N}}$ contains a subsequence denoted (to shorten notation) by $\left(l  \right)_{l \in \mathbb{N}}$ such that   
 $\left( L^{l}X^{l}_{t}    \right)_{l \in \mathbb{N}}$ 
 and $\left( C \widetilde{R}_{l} X^{l}_{t}    \right)_{l \in \mathbb{N}}$ are weakly convergent in 
 $L^{2}\left(\left( \Omega,\mathfrak{G}_{t}^{\xi ,W},\mathbb{P} \right), \mathfrak{h}\right)$. 
 By (\ref{3.41}),
 $$
 \left( \widetilde{R}_{l}X^{l}_{t}   ,  L^{l}X^{l}_{t} , C \widetilde{R}_{l} X^{l}_{t}   \right)
 \text{converges weakly in }
 L^{2}\left(\left( \Omega,\mathfrak{G}_{t}^{\xi ,W},\mathbb{P} \right), \mathfrak{h}^{3}\right) .
 $$
 
 The set 
$
\mathcal{D}\left( C\right) \times L  \left( \mathcal{D}\left( C\right) \right) \times C  \left( \mathcal{D}\left( C\right) \right) $
is closed in $\mathfrak{h}^{3}$,
because $L$ is relatively bounded with respect to $C$.
Then
the set of all triple $\left( \eta,A\eta ,L\eta\right) $,
with 
$\eta\in L_{C}^{2}\left( \left( \Omega,\mathfrak{G}_{t}^{\xi ,W},\mathbb{P} \right),\mathfrak{h}\right) $,
is a closed linear linear manifold of $L^{2}\left(  \left( \Omega,\mathfrak{G}_{t}^{\xi ,W},\mathbb{P} \right), \mathfrak{h}^{3}\right) $,
and hence  closed with respect to the weak topology of $L^{2}\left(\left( \Omega,\mathfrak{G}_{t}^{\xi ,W},\mathbb{P} \right), \mathfrak{h}^{3}\right)$.
Using (\ref{3.41}) we now get
$
\left( \widetilde{R}_{l}X^{l}_{t}  ,  L^{l}X^{l}_{t}  , C \widetilde{R}_{l} X^{l}_{t}    \right) 
$
converges weakly in $L^{2}\left(\left( \Omega,\mathfrak{G}_{t}^{\xi ,W},\mathbb{P} \right), \mathfrak{h}^{3}\right)$
to
$
\left( Z_{t}, L  Z_{t} ,C  Z_{t} \right)
$
as $l\rightarrow\infty$, 
which implies 
 \begin{equation}
\label{6.4}
L^{n_{k}}X_{t}^{n_{k}}  \longrightarrow_{k\rightarrow\infty} L Z_{t}
 \quad weakly\ in\ L^{2}\left(\left( \Omega,\mathfrak{G}_{t}^{\xi ,W},\mathbb{P} \right), \mathfrak{h}\right),
\end{equation}
and so (\ref{3.43}) holds by Condition H2.2. 
Taking $L = C$ in (\ref{6.4}) and using Lemma \ref{lema3} we get (\ref{3.42}).

Condition H2.1, together with (\ref{6.4}), shows that 
$G \left( t \right)   \widetilde{R}_{n_{k}}X_{t}^{n_{k}} $ 
converges
to $ G \left( t \right)   Z_{t} $   weakly in $L^{2}\left( \left( \Omega,\mathfrak{G}_{t}^{\xi ,W},\mathbb{P} \right), \mathfrak{h}\right)$ as $n \rightarrow \infty$.
It follows that for any  
$U \in L^{2}\left(\left( \Omega,\mathfrak{G}_{t}^{\xi ,W},\mathbb{P} \right), \mathfrak{h}\right)$,
\[
\mathbb{E} \left\langle U, \widetilde{R}_{n_{k}} G \left( t \right)   \widetilde{R}_{n_{k}} X^{n_{k}}_{t}   \right\rangle
=
\mathbb{E} \left\langle \widetilde{R}_{n_{k}}U, G \left( t \right)   \widetilde{R}_{n_{k}} X^{n_{k}}_{t}  \right\rangle  
\longrightarrow_{k\rightarrow\infty} 
\mathbb{E} \left\langle U, G \left( t \right)   Z_{t} \right\rangle .
\qedhere
\]
\end{proof}

By Lemma \ref{lema31},
as in Lemma 2.5 of   \cite{MoraReIDAQP}
we establish that  $Z_{t}\left( \xi \right)$ satisfies (\ref{eq:SSE}) a.s.
using the following predictable representation.

\begin{remark}
\label{nota3}
Let $\chi \in L^{2}\left(  \left( \Omega,\mathfrak{G}_{t}^{\xi,m} ,\mathbb{P} \right), \mathbb{C}\right)  $, with $t \in \left[ 0,T\right]$.  Then, there exist $ \left( \mathfrak{G} _{s}^{\xi ,m}\right)  _{s}$ - predictable processes $H^{1}, \cdots, H^{m}$ such that: 
(i)  $H^{1}, \cdots, H^{m} \in L^{2}\left( \left( \left[ 0,T\right] \times\Omega,dt\otimes\mathbb{P} \right), \mathbb{C}\right)  $; 
and
(ii) $ \chi=\mathbb{E}\left(  \chi|\mathfrak{G}_{0}^{\xi,m}\right)  +\sum_{j=1}^{m}\int_{0}^{t}H_{s}^{j}dW_{s}^{j} $.

\end{remark}

\begin{lemma}
\label{lema6}
Assume the setting of Theorem \ref{teorema1}.
Suppose that  $\left( X^{n_{k}}   \right)_{k \in \mathbb{N}}$ and $\chi$ 
are as in Lemma \ref{lema30} and Remark \ref{nota3} respectively. 
If $x \in \mathfrak{h}$,
then 
\[
 \lim_{k \rightarrow \infty}
\mathbb{E}  \left\langle   \chi x, 
\sum_{\ell=1}^{n_{k}} \int_{0}^{t}  L^{n_{k}}_{\ell} \left( s \right) X_{s}^{n_{k}}  dW_{s}^{j}
\right\rangle
 =
\mathbb{E} \left\langle  \chi x, \sum_{\ell=1}^{\infty}  
\int_{0}^{t}   L_{\ell}  \left( s \right) \pi \left(Z_{s} \right)  dW^{\ell}_{s} \right\rangle.
\]
 
\end{lemma}

\begin{proof} 

Throughout this proof, $H^{1}, \cdots, H^{m}$  are as in Remark \ref{nota3}.
First,
using Lemma \ref{lema1}, basic properties of stochastic integrals and Fubini's theorem we deduce
that for all  $n \geq m$,
\[
\mathbb{E}   \chi \left\langle  x, 
\sum_{\ell=1}^{n} \int_{0}^{t}  L^{n}_{\ell} \left( s \right) X_{s}^{n} dW_{s}^{\ell}
\right\rangle
=
 \sum_{\ell=1}^{m} \int_{0}^{t}   
 \mathbb{E}\, H^{\ell}_{s}  \left\langle  x, L^{n}_{\ell} \left( s \right) X_{s}^{n} \right\rangle 
 ds.
\]
By $\left\Vert \widetilde{R}_{n} \right\Vert \leq 1$ 
and
 $\widetilde{R}_{n} C  \subset C \widetilde{R}_{n}$,
combining (\ref{3.43}),  Lemmata \ref{lema1} and \ref{lema3}, and  the dominated convergence theorem
we obtain that for any $\ell=1,\ldots,m$,
\[
\int_{0}^{t}   \mathbb{E}\, H^{\ell}_{s}  \left\langle  x, L^{n_{k}}_{\ell} \left( s \right) X_{s}^{n_{k}}  \right\rangle ds
\longrightarrow_{k \rightarrow \infty}
 \int_{0}^{t}   \mathbb{E}\, H^{\ell}_{s}  \left\langle  x, L_{\ell}  \left( s \right) \pi \left( Z_{s} \right) \right\rangle ds,
\]
and so 
$$
\lim_{k \rightarrow \infty}
\mathbb{E}\,  \left\langle   \chi x, 
\sum_{\ell=1}^{n_{k}} \int_{0}^{t}  L^{n_{k}}_{\ell} \left( s \right)  X_{s}^{n_{k}} dW_{s}^{\ell}
\right\rangle
 =
\sum_{\ell=1}^{m}
 \int_{0}^{t}  
  \mathbb{E}\, H^{\ell}_{s}  \left\langle  x, L_{\ell} \left( s \right)  \pi \left( Z_{s} \right) \right\rangle 
 ds .
$$

Second, Lemmata \ref{lema1} and  \ref{lema3}, 
together with Condition H2.2,
yield
\[
 \sum_{\ell=1}^{m} \int_{0}^{t}  
  \mathbb{E}\, H^{\ell}_{s}  \left\langle  x, L_{\ell}  \left( s \right)  \pi \left(Z_{s} \right) \right\rangle ds
 =
\sum_{\ell=1}^{n}
 \mathbb{E}\, \chi  \int_{0}^{t}  \left\langle  x, L_{\ell}  \left( s \right) \pi \left(Z_{s} \right) \right\rangle dW^{\ell}_{s}
\]
whenever $n \geq m$.
Condition H2.1 and Lemma \ref{lema5} show that 
$
 \sum_{k=1}^{\infty }\left\Vert  L_{\ell} \left( t \right) y \right\Vert ^{2} 
\leq K\left( t \right)  \left\| y \right\|_{C}^{2}
$
for all $y$ in $\mathcal{D}\left( C \right)$ and $t \geq 0$.
Therefore 
$
\sum_{\ell=1}^{n}   \int_{0}^{t}  L_{\ell}  \left( s \right)  \pi \left(Z_{s} \right) dW^{\ell}_{s} 
$
converges in $L^{2}\left(\mathbb{P}, \mathfrak{h}\right)$ to
$
\sum_{\ell=1}^{\infty}   \int_{0}^{t}  L_{\ell}  \left( s \right)  \pi \left(Z_{s} \right) dW^{\ell}_{s}
$,
which implies that 
$$
 \sum_{\ell=1}^{m} \int_{0}^{t}  
  \mathbb{E}\, H^{\ell}_{s}  \left\langle  x, L_{\ell}  \left( s \right)  \pi \left(Z_{s} \right) \right\rangle ds
=
 \sum_{\ell=1}^{\infty} 
\mathbb{E}\, \chi  \int_{0}^{t}  \left\langle  x, L_{\ell}  \left( s \right)  \pi \left(Z_{s} \right) \right\rangle 
dW^{\ell}_{s} . 
\qedhere
$$
\end{proof}

\begin{lemma}
\label{lema32}
Adopt the assumptions of Theorem \ref{teorema1}. 
Let $T$ and $Z$ be defined as in Lemma \ref{lema30}. 
Then for all $t \in\left[  0,T\right]$ we have 
\begin{equation*}
Z_{t} =
\xi +\int_{0}^{t}G \left( s \right) \pi _{C}\left( Z_{s} \right) ds
+\sum_{\ell=1}^{\infty }\int_{0}^{t}L_{\ell} \left( s \right) \pi _{C}\left( Z_{s} \right) dW_{s}^{\ell} \quad a.s.
\end{equation*}
\end{lemma}

\begin{proof}

Consider $x \in \mathfrak{h}$ and 
 let $\left( X^{n_{k}} \right)_{k \in \mathbb{N}}$ be as in Lemma \ref{lema30}. 
According to Lemma \ref{lema6} we have 
\[
\lim_{k \to\infty}
\mathbb{E}  \left\langle   \chi x, 
\sum_{\ell=1}^{n_{k}} \int_{0}^{t}  L^{n_{k}}_{\ell} \left( s \right) X_{s}^{n_{k}} dW_{s}^{\ell}
\right\rangle
=
\mathbb{E} \left\langle  \chi x, \sum_{\ell=1}^{\infty}  \int_{0}^{t}   L_{\ell} \left( s \right) \pi \left(Z_{s} \right)  dW^{\ell}_{s} \right\rangle.
\]
Using Lemmata  \ref{lema1}, \ref{lema3} and \ref{lema31} 
and 
the dominated convergence theorem we obtain
\[
 \int_{0}^{t} \hspace{-1 pt} \mathbb{E} \left[ \left\langle  x, G^{n_{k}} \left( s \right)  X_{s}^{n_{k}}  
 \right\rangle
 \mathbb{E} \left[ \chi  |  \mathfrak{G} _{s}^{\xi ,W} \right] \right] ds 
\rightarrow_{k \rightarrow \infty} \hspace{-3 pt}
 \int_{0}^{t} \hspace{-2 pt} \mathbb{E} \left[ \left\langle  x, G \left( s \right)   \pi \left(Z_{s}  \right) \right\rangle 
 \mathbb{E} \left[ \chi  |  \mathfrak{G} _{s}^{\xi ,W} \right] \right] ds,
\]
since  
$\mathbb{E} \left[\,  \chi \, | \, \mathfrak{G} _{s}^{\xi ,W} \right] \in L^{2}\left(  \mathbb{P}, \mathbb{C}\right) $.
Thus,
combining (\ref{3.39}) with the definition of $X^{n}$ yields
\begin{equation}
\label{3.46}
  \mathbb{E} \chi \left\langle  x, Z_{t} \right\rangle 
  = 
 \mathbb{E} \chi  \left\langle  x, \xi +  \int_{0}^{t}  G \left( s \right)  \pi \left( Z_{s} \right) ds
  + \sum_{\ell=1}^{\infty}  
\int_{0}^{t}  \left\langle  x, L_{\ell} \left( s \right)  \pi \left(Z_{s}  \right) \right\rangle dW^{\ell}_{s} \right\rangle.
\end{equation}
As in  \cite{MoraReIDAQP}, using  a monotone class theorem 
(e.g., Th. I.21 of  \cite{Dellacherie}) 
we extend the range of validity of  (\ref{3.46}) from 
$\chi \in L^{2}\left(  \left( \Omega,\mathfrak{G}_{t}^{\xi,m} ,\mathbb{P} \right), \mathbb{C}\right) $ 
to any bounded $\chi \in L^{2}\left(  \left( \Omega,\mathfrak{G}_{t}^{\xi,W} ,\mathbb{P} \right), \mathbb{C}\right) $, 
which completes the proof.
\end{proof}

We are now in a position to finish the proof of Theorem \ref{teorema1} by classical arguments. 

\begin{proof}[Proof of Theorem \ref{teorema1}]
 Consider $T >0$. 
 First, we combine Lemma \ref{lema5} with It\^o's formula to deduce that there exists at most one strong $C$-solution of (\ref{eq:SSE}) on $\left[ 0,T\right] $ (see proof of Lemma 2.2 of  \cite{MoraReIDAQP} for details). Second, for all 
$t \in \left[ 0,T\right]$, we set  
 \[
 Z^{T}_{t}= 
 \xi +\int_{0}^{t}G \left( s \right)  \pi _{C}\left( Z_{s} \right) ds
 +\sum_{\ell=1}^{\infty }\int_{0}^{t}L_{\ell} \left( s \right)  \pi _{C}\left( Z_{s} \right) dW_{s}^{\ell},
 \]
where $Z$ is as in Lemma \ref{lema30}.
Using Lemma \ref{lema32} we see that $ Z^{T}$ is a continuous version of $Z$. 
Hence $ Z^{T}$ is a strong $C$-solution of (\ref{eq:SSE}) on $\left[ 0,T\right] $, 
and so $ Z^{T}$ is the unique one.
 
 Define $\widetilde{\Omega}$ to be the set of all $\omega$ satisfying $Z^{n}_{t}  \left( \omega \right) = Z^{n+1}_{t}  \left( \omega \right)$ for all $n \in \mathbb{N}$ 
and any $t \in \left[ 0,n\right]$. 
For any $t  \in  \left[ 0,n\right]$ with $n \in \mathbb{N}$, 
we choose  
$ X_{t}\left( \xi \right) \left( \omega \right) = Z^{n}_{t}  \left( \omega \right)$
whenever $ \omega \in \widetilde{\Omega}$.
Set $X\left( \xi \right) \equiv 0$ in the complement of  $\widetilde{\Omega}$.
 Thus $X\left( \xi \right)$ is the unique strong $C$-solution of (\ref{eq:SSE}) on $\left[ 0, \infty  \right[$.
\end{proof}

\subsection{Proof of Lemma \ref{lema:EquivH2.3}}     
\label{subsec:lemaEquivH2.3} 
\begin{proof}

According to Lemma \ref{lema:Measurability}
we have that 
$ G \left( t \right) \circ\pi_{C^2} $ and $C L_{\ell} \left( t \right) \circ\pi_{C^2} $
are measurable functions from 
$ \left(\mathfrak{h}, \mathcal{B} \left( \mathfrak{h} \right) \right)$
to $\left(\mathfrak{h}, \mathcal{B} \left( \mathfrak{h} \right) \right)$,
and so  $\mathcal{L}^{+} \left( t, \zeta \right) $ is a positive random variable
for every  $\zeta \in  L_{C}^{2}\left( \mathbb{P},\mathfrak{h}\right) \bigcap \mathfrak{D}_{1}$.
Hence $\mathbb{E} \left( \mathcal{L}^{+} \left( t, \zeta \right)  \right)$ is well-defined.

Condition H2.3 leads straightforward to Condition H2.3'.
In the other direction,
we assume from now on that H2.3' holds.
Fix $t \geq 0$. 
To obtain a contradiction,
suppose that for any $n \in \mathbb{N}$ there exists $x_n \in \mathfrak{D}_{1}$ such that
\begin{equation}
\label{eq:2.1}
 2\Re\left\langle C^{2} x_n, G \left( t \right) x_n\right\rangle 
+\sum_{\ell=1}^{\infty }\left\Vert C L_{\ell} \left( t \right) x_n \right\Vert ^{2}
> n  \left\Vert x_n \right\Vert_{C}^{2} .
\end{equation}
We can consider a random variable $\zeta$ defined by
$$
\mathbb{P} \left( \zeta = y_n \right)
=
\frac{p}{n^2 \left( 1 + \left\Vert C y_n \right\Vert ^2 \right)} ,
$$
where $y_n = x_n / \left\Vert x_n \right\Vert $ and
$$
1/p =  
\sum_{n=1}^{\infty} \frac{1}{  n^2 \left( 1 + \left\Vert C y_n \right\Vert ^2 \right) }  < \infty.
$$
Then  $\left\Vert  \zeta \right\Vert = 1$,
$\zeta \in \mathfrak{D}_{1}$,
and
$
\mathbb{E} \left\Vert \zeta \right\Vert_{C}^2
=
p \sum_{n=1}^{\infty} 1/n^2
<
\infty .
$ 
Using \eqref{eq:2.1} yields 
$$
\mathbb{E}  \mathcal{L}^{+} \left( t, \zeta \right) 
\geq
p \sum_{n=1}^{\infty} 1/n 
= \infty ,
$$
which contradicts Condition H2.3'.
Therefore, 
there exists a constant $\beta \left( t \right) \geq 0$
such that for all $x \in \mathfrak{D}_{1}$,
\begin{equation}
\label{eq:2.2}
 2\Re\left\langle C^{2} x, G \left( t \right) x \right\rangle 
+\sum_{\ell=1}^{\infty }\left\Vert C L_{\ell} \left( t \right) x \right\Vert ^{2}
\leq
 \beta \left( t \right)  \left\Vert x \right\Vert_{C}^{2} .
\end{equation}
By abuse of notation, we denote by  $\beta \left( t \right)$
the smallest $\beta \left( t \right)$ satisfying \eqref{eq:2.2}.

Suppose, contrary to H2.3,  that  
$\sup_{t \in \left[ 0 , T \right]} \beta \left( t \right) = \infty$
for some $T>0$.
Then,
there exists a sequence $\left( s_n \right)_{n \in \mathbb{N}}$ 
of different elements of $ \left[ 0 , T \right]$ such that 
\begin{equation}
\label{eq:2.3}
 2\Re\left\langle C^{2} z_n, G \left( s_n \right) z_n \right\rangle 
+\sum_{\ell=1}^{\infty }\left\Vert C L_{\ell} \left( s_n\right) z_n \right\Vert ^{2}
>
 n^3  \left\Vert z_n \right\Vert_{C}^{2} 
\end{equation}
for some $z_n  \in \mathfrak{D}_{1}$ satisfying $\left\Vert  z_n \right\Vert = 1$.
Similarly to the paragraph above,
we can choose a random variable $\zeta$ defined by
$$
\mathbb{P} \left( \zeta = z_n \right)
=
\frac{p}{n^2 \left( 1 + \left\Vert C z_n \right\Vert ^2 \right)} ,
$$
with 
$
1/p =  
\sum_{n=1}^{\infty} 1 / \left(  n^2 \left( 1 + \left\Vert C z_n \right\Vert ^2 \right) \right)  < \infty
$.
Hence
$
\mathbb{E} \left\Vert \zeta \right\Vert_{C}^2
<
\infty 
$ .
From \eqref{eq:2.3} we deduce that
$$
\mathbb{E} \mathcal{L}^{+} \left( s_n, \zeta \right)  
\geq
 \mathcal{L}^{+} \left( s_n, z_n \right)  
\geq
n p ,
$$
which contradicts Condition H2.3'.
Taking 
$$
\alpha \left( t \right) = \sup_{t \in \left[ 0 , t \right]} \beta \left( s \right) < \infty.
$$
we can assert that Condition H2.3 holds. 
\end{proof}

\subsection{Proof of Theorem \ref{teo:Smooth}}     
\label{subsec:teoSmooth}

Throughout this subsection,
$C$ denotes the operator in $L^2 \left( \mathbb{R}^d, \mathbb{C} \right)$ given by 
$C=  - {\Delta} + \left| x \right| ^{2}$.
Moreover,
$ \left\| \cdot \right\| $ stands for the norm in $L^2 \left( \mathbb{R}^d, \mathbb{C} \right)$,
and 
we shall often use Einstein summation convention 
(each index can appear at most twice in any term, repeated indexes are implicitly summed over).

Since $\left| x \right| ^{2}$ is locally in $L^{2}\left( \mathbb{R}^{d} ,\mathbb{C}\right) $,
the operator $ - {\Delta} + \left| x \right| ^{2}$ is  
essentially self-adjoint on $C^{\infty}_{c}\left( \mathbb{R}^{d} , \mathbb{C} \right) $ 
(see, e.g., Th. X.29 of  \cite{ReedSimonVol2}). 
The Hermite functions (i.e., Hermite polynomials multiplied by $\hbox{\rm e}^{-x^2/2}$)  
are the eigenfunctions of the operator in $L^{2}\left( \mathbb{R},\mathbb{C}\right) $ 
given by $  -  d^{2} / dx^{2}  + x^{2} $,
and hence 
the Schwarz space of rapidly decreasing functions is an essential domain for 
 $\left( -  d^{2} / dx^{2}  + x^{2} \right)^2$.
We can now use standard approximation arguments to show that 
$C^{\infty}_{c}\left( \mathbb{R} , \mathbb{C} \right) $ is a core for 
$\left( -  d^{2} / dx^{2}  + x^{2} \right)^2$,
which implies that 
$C^{\infty}_{c}\left( \mathbb{R}^{d} , \mathbb{C} \right) $ is a core for  $C^2$
(see, e.g., Th. VIII.33 of  \cite{ReedSimonVol1}),
and so $C^{\infty}_{c}\left( \mathbb{R}^{d} , \mathbb{C} \right) $ is 
an essential domain for $C$.
As 
$- {\Delta} + \left| x \right| ^{2} 
= \sum_{j=1}^d (\partial_j + x_j)^*(\partial_j+x_j) + d I$,
the operator $C$ is bounded from below by $d$ times the identity operator $I$.


We next provide some relative bounds on  $C$,  $\Delta$ and the multiplication operator by $\left|x \right|^2$.

\begin{lemma}\label{lemma:key-op-ineq}
Let  $f \in  C^{\infty}_{c}\left( \mathbb{R}^{d} , \mathbb{C}\right) $.  Then
\begin{equation}
  \label{eq:C-Delta-x^2-identity}
      \left\| C f\right\|^2 
    = \left\| \Delta f\right\|^2 + \left\|\, |x|^2 f\right\|^2 
 + 
2\sum_{j=1}^d\left\langle \partial_j f, |x|^2\partial_j f\right\rangle 
-2d\left\| f\right\|^2,
\end{equation}
\begin{equation}
  \label{eq:(1+|x|)gradf-estimate} 
\sum_{j=1}^d \left\|  (1+|x|)\partial_j f\right\|^2  \le  4 \left\| C f\right\|^2, 
\end{equation}
\begin{equation}
                         \label{eq:(1+|x|^2)-estimate}
\left\|\, (1+|x|^2) f\right\|^2  \le  8 \left\| C f\right\|^2.
\end{equation}
\end{lemma}

\begin{proof}  Using integration by parts yields
\begin{equation}
\label{eq:C-Delta-x^2-identity-a1}
 \left\| C   f   \right\| ^{2} 
 =  \sum_{j,k=1}^d \left( \left\langle -\partial_j^2 f, -\partial_k^2 f \right\rangle    
+\left\langle x_j^2 f, x_k^2 f \right\rangle \right) 
 -  \sum_{j,k=1}^d \left\langle f, (\partial_j^2 x_k^2 + x_j^2\partial_k^2)f \right\rangle .
\end{equation}
A short computation based on the commutation relation 
$[ \partial_j,x_j ] = I$
gives
\begin{equation}
\label{eq:C-Delta-x^2-identity-a2}
 (\partial_j^2 x_j^2 + x_j^2\partial_j^2)f
= 2\partial_j x_j^2 \partial_j f + 2(\partial_j x_j - x_j\partial_j) f
= 2\partial_j x_j^2 \partial_j f + 2 f.
\end{equation}
For any  $j\not=k$ we have
$ 
\partial_j^2 x_k^2 +  x_j^2\partial_k^2
=
\partial_j x_k^2 \partial_j  +  \partial_k x_j^2\partial_k
$
on
$C^{\infty}_{c}\left( \mathbb{R}^{d} ,\mathbb{C}\right) $,
which together with (\ref{eq:C-Delta-x^2-identity-a1}) and (\ref{eq:C-Delta-x^2-identity-a2}) 
implies  (\ref{eq:C-Delta-x^2-identity}).

We now check inequality (\ref{eq:(1+|x|)gradf-estimate}). 
Combining (\ref{eq:C-Delta-x^2-identity}) with the inequality
\[
\sum_{j=1}^d \left\|  (1+|x|)\partial_j f\right\|^2 
 \le  2\sum_{j=1}^d \left\langle \partial_j f, (1+|x|^2)\partial_j f\right\rangle 
= 2 \left\langle  f, -\Delta f\right\rangle 
+  2\sum_{j=1}^d \left\langle \partial_j f,|x|^2\partial_j f\right\rangle,
\]
we obtain
$
\sum_{j=1}^d \left\|  (1+|x|)\partial_j f\right\|^2 
 \le 
 2 \left\langle  f, -\Delta f\right\rangle - \left\|\Delta f\right\|^2   
 -\left\|\, |x|^2 f\right\|^2+\left\| Cf\right\|^2 +2d \left\| f\right\|^2
 $,
and hence
\begin{equation}
\label{eq:(1+|x|)gradf-estimate-a1}
 \sum_{j=1}^d \left\|  (1+|x|)\partial_j f\right\|^2 
 \le 
 -\left\| (\Delta+1) f\right\|^2 +\left\| Cf\right\|^2 +(2d+1) \left\| f\right\|^2 .
\end{equation}
By  $C^2\ge d^2 1$,
 $(2d+1) \left\| f\right\|^2\le (2d^{-1}+d^{-2})\left\| Cf\right\|^2\le 3\left\| Cf\right\|^2$ since $d\ge 1$.
 Then,  (\ref{eq:(1+|x|)gradf-estimate}) follows from (\ref{eq:(1+|x|)gradf-estimate-a1}).

According to (\ref{eq:C-Delta-x^2-identity}), we have 
\[
\left\| (1+|x|^2)  f\right\|^2 
 \le  2 \left\| f\right\|^2  + 2\left\| |x|^2  f\right\|^2 
 \le   2(2d+1) \left\| f\right\|^2 +  2 \left\| C f\right\|^2  .
\]
Then
$
\left\| (1+|x|^2)  f\right\|^2
 \le  2(1+2d^{-1}+d^{-2}) \left\| C f\right\|^2
$,
which leads to (\ref{eq:(1+|x|^2)-estimate}).
\end{proof}

\begin{lemma}\label{lem:G-satisfies-H.2.1}
Under Hypothesis \ref{hyp:growth-cond}, 
the operators $G(t)$ and $L_\ell(t)$ satisfy 
Hypothesis \ref{hyp:L-G-C-domain},
as well as Conditions H2.1 and H2.2 of Hypothesis \ref{hyp:CF-inequality}.
\end{lemma}

\begin{proof} For all $f\in C^{\infty}_{c}\left( \mathbb{R}^{d} ,\mathbb{C} \right) $ we have 
(Einstein summation convention on $j$)
\[
\left\| H(t) f \right\| \le \alpha \left\|\Delta f \right\| 
+ 2 \left\| A^j(t,\cdot)\partial_j f \right\|  
+  \left\| f \partial_j A^j (t,\cdot) \right\| 
+ \left\| V(t,\cdot)f\right\|.
\]
Combining the Schwarz inequality with (\ref{eq:(1+|x|)gradf-estimate}) gives
\[
  \left\| A^j(t,\cdot)\partial_j f \right\|   
 \le K(t)\left(d\, \textstyle{\sum_{j=1}^d}\left\|\left (1+|x|\right)\partial_j f \right\|^2\right )^{1/2} 
   \le  2d^{1/2} K(t) \left\| Cf\right\|. 
\]
By (\ref{eq:(1+|x|^2)-estimate}), 
$\left\| V(t,\cdot)f\right\| \le 8^{1/2}K(t)\left\| Cf\right\|$.
Moreover, 
Condition H4.1 implies 
$$
\left\| f \partial_j A^j (t,\cdot) f \right\| \le d\, K(t)  \left\| f\right\|, 
$$
and 
the identity  (\ref{eq:C-Delta-x^2-identity}) yields
$ \left\|\Delta f \right\|  \le 
 \left\| C f \right\| +(2  d)^{1/2} \left\| f \right\|$.
 Therefore
\begin{equation}\label{eq:H-C-rel-bd}
\left\| H(t) f \right\| \le 
\left(  \alpha +\left (2d^{1/2}+8^{1/2}\right)K(t)\right)\left\| Cf\right\|
+ \left(\alpha (2d)^{1/2}+dK(t)\right)  \left\| f\right\|.
\end{equation}

A straightforward computation yields 
\[
 L_\ell(t)^* L_\ell(t)  
 = 
 - \overline{\sigma}_{\ell j} \sigma_{\ell k}\partial_j\partial_k  
 - \left(
                 \overline{\sigma}_{\ell k}\eta_\ell -\overline{\eta}_\ell\sigma_{\ell k}
                 + \overline{\sigma}_{\ell j} \mlt{  \partial_j \sigma_{\ell k} }
       \right)\partial_k
  +\left(\overline{\eta}_\ell\eta_\ell
           - \overline{\sigma}_{\ell j} \mlt{ \partial_j\eta_\ell  }
           \right) 
\]
Since $1\le \ell\le m$,
using Condition H4.2 and the Schwarz inequality we deduce that
\[
\left\|\overline{\sigma}_{\ell j} 
\sigma_{\ell k}\partial_j\partial_k f  \right\| 
\le m d\, K(t)^2 \left(\textstyle{\sum_{j,k=1}^d} 
\left\|\partial_j\partial_k f\right\|^2\right)^{1/2}
= m d\, K(t)^2 \left\|\Delta f\right\| .
\]
From $|\eta_\ell|\le K(t)(1+|x|)$, $|\sigma_{\ell j}|\le K(t)$ and $|\partial_j\sigma_{\ell j}|\le K(t)$
we have 
\begin{eqnarray*}
\left\| \left(\overline{\eta}_\ell\eta_\ell
- \overline{\sigma}_{\ell j} \mlt{ \partial_j\eta_\ell } \right)f \right\|
& \le &   
        \left\|  \overline{\eta}_\ell\eta_\ell f \right\|
        +  \left\|  \overline{\sigma}_{\ell j} \mlt{  \partial_j\eta_\ell  }  f  \right\| \\
& \le &  2m K(t)^2 \left\| (1+|x|^2)f\right\| + 2md K(t)^2 \left\| (1+|x|)f\right\|\\
& \le &  2m(2d+1) K(t)^2 \left\| (1+|x|^2)f\right\|. 
\end{eqnarray*}
Similarly, 
combining the Schwarz inequality with (\ref{eq:(1+|x|)gradf-estimate}) yields
\begin{align*}
  \left\| \left(
\overline{\sigma}_{\ell j}\  \mlt{    \partial_j\sigma_{\ell k}  }
             -\overline{\eta}_\ell\sigma_{\ell k}
            + \eta_\ell\overline{\sigma}_{\ell k}\right)\partial_kf \right\|
 \ & \le 4m  K(t)^2  
\sum_{k=1}^d  \left\|  (1+|x|)\partial_k f\right\|   
  \\
 & \le  8md^{1/2}\, K(t)^2 \left\| Cf\right\| .
\end{align*}
 Summing up,
 $\sum_{\ell=1}^m \left\| L_\ell(t)^*L_\ell(t) f \right\|$ 
 is less than or equal to  $m K(t)^2$ times
\begin{eqnarray*} 
& & 8 d^{1/2}\left\|Cf\right\| + d \left\|\Delta f\right\| 
+ 2(2d+1)\left\| |x|^2f\right\| + 2(2d+1)\left\| f\right\| \\
&\le  & 8d^{1/2}\left\|Cf\right\| 
+ 2 (2d+1)\left( \left\|\Delta f\right\|+\left\| |x|^2f\right\| \right)+ 2(2d+1)\left\| f\right\| \\
&\le  & 8d^{1/2}\left\|Cf\right\| 
 + 2 (2d+1)\left( \left\| C f\right\|+(2d)^{1/2} \left\| f\right\|  \right)+ 2(2d+1)\left\| f\right\|.
 \end{eqnarray*}
This, together with  (\ref{eq:H-C-rel-bd}),  shows that $G(t)$ satisfies Condition H2.1
since $C^\infty_c(\mathbb{R}^d;\mathbb{C})$ is a core for $C$.
In a similar, but simpler, way (we deal now with first order 
differential operators), we can prove that the operators 
$L_\ell(t)$ satisfy Condition H2.2.

Let $g \in C^\infty_c(\mathbb{R}^d;\mathbb{C})$.
Since $ \sigma_{\ell k}\left(t,\cdot\right) $ is continuous, 
using Fubini's theorem we deduce the measurability of 
$
t \mapsto  \left\langle   \phi, \overline{\sigma_{\ell k}\left(t,\cdot\right)} g \right\rangle    
$
for all $\phi \in C^\infty_c(\mathbb{R}^d;\mathbb{C})$,
and so $t \mapsto  \overline{\sigma_{\ell k}\left(t,\cdot\right)} g$ is measurable.
Combining Lemma \ref{lema:Measurability} with (\ref{eq:(1+|x|)gradf-estimate})
yields the measurability of $ f \mapsto \partial_k \pi_C \left( f \right) $
as a map from $L^2 \left( \mathbb{R}^d, \mathbb{C} \right)$ to $L^2 \left( \mathbb{R}^d, \mathbb{C} \right)$.
Therefore
$
\left(t, f \right)
\mapsto
 \left\langle \overline{ \sigma_{\ell k}\left(t,\cdot\right)} g  ,  \partial_k \pi_C \left( f \right) \right\rangle    
$
is measurable,
which implies the measurability of 
$
\left(t, f \right)
\mapsto
 \sigma_{\ell k}\left(t,\cdot\right) \partial_k \pi_C \left( f \right) 
$
as a function from 
$ \left[ 0, \infty \right[ \times L^2 \left( \mathbb{R}^d, \mathbb{C} \right)$
to  $L^2 \left( \mathbb{R}^d, \mathbb{C} \right)$.
In the same manner we can see that
$
\left(t, f \right)
\mapsto
 \eta_{\ell k}\left(t,\cdot\right) \pi_C \left( f \right) 
$
is measurable,
hence Condition H1.1 holds.
Similarly, 
we can obtain that $G(t)$ satisfies Condition H1.2. 
\end{proof}

We now verify Condition H2.3 with 
$\mathfrak{D}_1 = C^\infty_c\left( \mathbb{R}^{d} ,\mathbb{C}\right) $,
which is the most complicated step of our proof.
The key inequality in Condition H2.3 at a formal purely algebraic level reads as 
$\mathcal{L}(C^2)\le \alpha(t) \left( C^2 + I \right)$,
where $\mathcal{L}$ is the formal time-dependent Lindbladian associated 
with the operators $G(t)$ and $L_\ell(t)$, namely 
$\mathcal{L}(X)= G(t)^*X + \sum_\ell L_\ell(t)^* X L_\ell(t) 
+ X G(t)$. Decomposing $\mathcal{L}$ as the sum of 
a Hamiltonian part $i[H(t),\cdot]$ and a dissipative part 
$\mathcal{L}_0(X) =\mathcal{L}(X)-i[H(t),X]$ 
we check separately that $i[H(t),C^2]\le K(t) \left( C^2 + I \right)$ 
(Lemma \ref{lem:comm-H-C-estim}) and 
$\mathcal{L}_0(C^2)\le K(t) \left( C^2  + I \right)$ 
(Lemmata \ref{lemma:CCP-part-C}, \ref{lem:estimate-L0-step1}
and \ref{lem:estimate-L0}) in the quadratic form sense. 

\begin{lemma}\label{lem:comm-H-C-estim}
Suppose that Hypothesis \ref{hyp:growth-cond}  holds. 
Then,
for all $f\in C^\infty_c\left( \mathbb{R}^{d} ,\mathbb{C}\right)$ we have
\begin{equation}\label{eq-H-C-estimate}
2\Re \left\langle C^2 f, iH(t)f\right\rangle 
\le K(t) \left(  \left\| C f\right\|^2 + \left\| f\right\|^2 \right).
\end{equation} 
\end{lemma}

\begin{proof}
Since $f\in C^\infty_c\left( \mathbb{R}^{d} ,\mathbb{C}\right)$,
\begin{eqnarray*}
&&
\left\langle  C^2f, iH(t)f\right\rangle + \left\langle  iH(t)f,   C^2f\right\rangle  
\\
& = & \left\langle  Cf, i (H(t)C+[C,H(t)])f\right\rangle 
+ \left\langle  i(H(t)C +[C,H(t)]) f, Cf\right\rangle 
 \\
& = & - i \left\langle Cf, [H(t),C]f\right\rangle + i \left\langle [H(t),C]f, Cf\right\rangle .
\end{eqnarray*}
Then
$
\left| 2\Re \left\langle C^2 f, iH(t)f\right\rangle \right|
=
\left| 2\Im \left\langle Cf, [H(t),C]f\right\rangle  \right|
\leq
2\left\|Cf\right\|\cdot\left\| [H(t),C]f\right\|
$.
A computation allows us  to write the commutator  $[H(t),C]$  as
\begin{eqnarray*}
& & 4 i \mlt{ \partial_k A^j } \partial_k\partial_j 
+ \left( 2   \mlt{\partial_j V }  
 -   4\alpha x_j + 2    i\mlt{\Delta A^j} 
+i\mlt{ \partial_j\partial_k A^k } \right)\partial_j
\\
&  & 
 + \left( \mlt{ \Delta V } -  2\alpha d 
+i \mlt{ \partial_j \Delta A^j } 
  +   4i x_j A^j \right) .
\end{eqnarray*}
Using Condition H4.1 and the Schwarz inequality we obtain that
the norm of the first term, acting on a function $f$, is less than or equal to 
\[
4 K(t) \textstyle{\sum_{j,k=1}^d} \left\| \partial_k\partial_j f \right\|
\le 4 K(t) d \left(   \textstyle{\sum_{j,k=1}^d} \left\| \partial_k\partial_j f \right\|^2\right)^{1/2}
= 2K(t) d \left\| \Delta f \right\|.
\]
The norm of the second term, with first order partial derivatives of $f$, is upper bounded by
$
K\left(t \right) \sum_{k=1}^d  \left\|(1+|x|)\partial_k f\right\| 
$,
which is less than or equal to 
$$
d^{1/2} K(t)  \left( \sum_{k=1}^d  \left\|(1+|x|)\partial_k f\right\| ^2 \right)^{1/2} \mspace{-2 mu} .
$$
The third term is not bigger than $ K(t) \left\| (1+|x|^2) f\right\|  $.
We now use Lemma \ref{lemma:key-op-ineq} to get (\ref{eq-H-C-estimate}).
\end{proof}

Starting from the formal algebraic equality 
\[
\mathcal{L}_0(C^2) = C\mathcal{L}_0(C)+\mathcal{L}_0(C) C 
+\sum_{\ell=1}^d \left[C,L_\ell(t)\right]^*\left[C,L_\ell(t)\right] 
\]
written in the quadratic form sense,
we now establish 
an estimate of 
$\left| \langle f, \mathcal{L}_0(C^2)f\rangle \right|$
(formally the left-hand side of the inequality (\ref{4.1}) given below).
Note that $\mathcal{L}_0(C^2)$ 
does not make sense as a sixth (or fourth, after simplifications) order 
differential operator acting on 
$f \in  C^\infty_c\left( \mathbb{R}^{d} ,\mathbb{C}\right) $ 
because $\sigma_{\ell k}, \eta_\ell$ are only three-times 
differentiable,
but $\mathcal{L}_0(C)$ does. The right-hand side of  (\ref{4.1}),  
however, can be written rigorously as
$
\sum_{\ell=1}^m \left\| \left[ C, L_\ell \right] f  \right\|^{2} 
+ 2 \left\| Cf  \right\|\cdot \left\|\mathcal{L}_0(C)f  \right\|
$.

\begin{lemma}\label{lemma:CCP-part-C} 
For all 
$f \in  C^\infty_c\left( \mathbb{R}^{d} ,\mathbb{C}\right) $ we have
\begin{eqnarray}
\label{4.1}
&&  \sum_{\ell=1}^m \left(\left\| C L_\ell f \right\|^{2}
-  \Re\left\langle C^{2} f, L^*_\ell L_\ell f \right\rangle \right) 
 \leq 
\sum_{\ell=1}^m \left\| \left[ C, L_\ell \right] f  \right\|^{2} 
 \\
 \nonumber
 && \hspace{5cm}
 + 
\left\| Cf  \right\| 
\left\| \sum_{\ell=1}^m \left( L_\ell^* \left[ C, L_\ell\right]  +  \left[ L_\ell^*, C\right] L_\ell\right) f  \right\|.
\end{eqnarray}
\end{lemma}

\begin{proof}
Rearranging terms  we have that for all $\ell=1,\dots,m$,
\begin{eqnarray*}
&  & 
   \left\langle C  L_\ell f, C  L_\ell f \right\rangle 
 - \left\langle C^{2}  f, L_\ell^*L_\ell f\right\rangle
\\
 & = & 
   \left\langle \left(  L_\ell C +  \left[ C, L_\ell \right] \right)  f,  \left(  L_\ell C +  \left[ C, L_\ell \right]\right)  f \right\rangle 
  - \left\langle C  f, (L_\ell^* C +[C,L_\ell^*]) L_\ell  f\right\rangle  \\
 & = & 
  \left\langle  L_\ell C  f,  L_\ell C  f\right\rangle 
+ \left\langle  L_\ell C  f,   \left[ C, L_\ell \right]f \right\rangle 
+ \left\langle \left[ C, L_\ell \right]  f,   L_\ell C   f \right\rangle 
+  \left\langle   \left[ C, L_\ell \right]   f,    \left[ C, L_\ell \right]  f \right\rangle  \\
 &  & -
 \left\langle L_\ell C  f,  C   L_\ell  f\right\rangle
- \left\langle C  f, [C,L_\ell^*] L_\ell  f\right\rangle
\end{eqnarray*}
Note that the sum of the first, second and fifth term vanishes and the third is equal to 
$ \left\langle L_\ell^*\left[ C, L_\ell \right]  f,   C   f \right\rangle$. 
We find then
\[
 \left\langle C  L_\ell f, C  L_\ell f \right\rangle 
 -  \left\langle C^{2}  f, L_\ell^*L_\ell  f\right\rangle
=  \left\|  \left[ C, L_\ell \right]  f \right\|^2 
+ \left\langle L_\ell^*\left[ C, L_\ell \right]  f,   C   f \right\rangle 
+ \left\langle C  f, [L_\ell^*,C] L_\ell  f\right\rangle,
\]
and so taking the real part we can write 
\begin{align*}
\left\langle C  L_\ell f, C  L_\ell f \right\rangle 
 -  \Re\left\langle C^{2}  f, L_\ell^*L_\ell  f\right\rangle 
\ & =   
\left\| \left[ C, L_\ell \right]  f \right\|^2 
+ \frac{1}{2} \left\langle L_\ell^*\left[ C, L_\ell \right]  f,   C   f \right\rangle 
\\
& \hspace{-2.5cm} 
 + \frac{1}{2} \left\langle C f,    L_\ell^*\left[ C, L_\ell \right]   f \right\rangle  
 +  \frac{1}{2}\left(\left\langle C  f, [L_\ell^*,C] L_\ell  f\right\rangle
 + \left\langle [L_\ell^*,C] L_\ell  f, C f\right\rangle\right) ,
\end{align*}
which implies 
\begin{align*}
\left\langle C  L_\ell f, C  L_\ell f \right\rangle 
 -  \Re\left\langle C^{2}  f, L_\ell^*L_\ell  f\right\rangle 
\ & =    \left\|\left[ C, L_\ell \right]  f  \right\|^2 
+ \frac{1}{2} \left\langle L_\ell^*\left[ C, L_\ell \right]  f,   C   f \right\rangle 
 \\
& 
 \hspace{-2.5cm}
+ \frac{1}{2} \left\langle [L_\ell^*,C] L_\ell  f, C f\right\rangle 
 + \frac{1}{2}\left(\left\langle C  f, [L_\ell^*,C] L_\ell  f\right\rangle
 + \left\langle C f,    L_\ell^*\left[ C, L_\ell \right]   f \right\rangle\right).
\end{align*}
The conclusion follows  summing up over $\ell$ and applying the Schwarz inequality.
\end{proof}

We now show that $\mathcal{L}_0(C)$ is a second 
order differential operator with well-behaved coefficients 
allowing us to prove that $\mathcal{L}_0(C)$ is relatively bounded with respect to $C$.

\begin{lemma}
\label{lem:estimate-L0-step1}
Under the Hypothesis \ref{hyp:growth-cond}, 
for all $f\in C^\infty_c(\mathbb{R}^d, \mathbb{C})$ we have 
\begin{equation}\label{eq:estimate-norm-L0(C)}
\left\Vert \mathcal{L}_0(C) f \right\Vert 
\le K(t) \left\Vert Cf\right\Vert .
\end{equation}
\end{lemma}

\begin{proof} 
Simple algebraic computations yield 
\begin{eqnarray*}
2\mathcal{L}_0(x_j) 
 =  
\left( L_\ell^*[x_j,L_\ell] +[L_\ell^*,x_j]L_\ell\right)
& = &
  \left(\partial_k\overline{\sigma}_{\ell k}-\overline{\eta}_\ell\right)\sigma_{\ell j}
- \overline{\sigma}_{\ell j}\left(\sigma_{\ell k}\partial_k 
+ \eta_\ell\right)\\
&=& 
 \partial_k(\sigma^*\sigma)_{k j} - (\sigma^*\sigma)_{j k}\partial_k
-\overline{\eta}_\ell \sigma_{\ell j}- \overline{\sigma}_{\ell j} \eta_\ell,
\end{eqnarray*}
which implies
$
2\mathcal{L}_0(x_j) 
=
\left( (\sigma^*\sigma)_{k j}- (\sigma^*\sigma)_{j k} \right) \partial_k
+\mlt{ \partial_k(\sigma^*\sigma)_{k j} } - \overline{\eta}_\ell \sigma_{\ell j}-\eta_\ell\overline{\sigma}_{\ell j}
$.
From
$
\mathcal{L}_0(|x|^2)  =  x_j \mathcal{L}_0(x_j) 
    + \mathcal{L}_0(x_j)  x_j + [ x_j, L_\ell^*][L_\ell, x_j] 
 $
 it follows that 
\begin{equation}
 \label{4.3}
\mathcal{L}_0(|x|^2)  
 =  
 x_j\left( (\sigma^*\sigma)_{k j} 
-  (\sigma^*\sigma)_{k j}\right)\partial_k
+x_j\mlt{\partial_k(\sigma^*\sigma)_{k j}}
-2\Re  (\overline{\eta}_{\ell}\sigma_{\ell j}) x_j
+(\sigma^*\sigma)_{j j} .
\end{equation}

In a similar way $-\Delta=\sum_j -\partial_j\partial_j$, and 
\begin{eqnarray*}
 2\mathcal{L}_0(\partial_j) 
 & = & 
  L_\ell^*[\partial_j,L_\ell] +[L_\ell^*,\partial_j]L_\ell
\\
  &   =   &
\left(-\partial_h\overline{\sigma}_{\ell h}+\overline{\eta}_\ell\right) 
\left(\mlt{ \partial_j\sigma_{\ell k}  }\partial_k 
+ \mlt{ \partial_j\eta_\ell } \right)
\\ 
&&
+\left( \partial_h \mlt{ \partial_j\overline{\sigma}_{\ell h}  }
- \mlt{ \partial_j\overline{\eta}_\ell} \right)
\left(\sigma_{\ell k}\partial_k + \eta_\ell\right) .
\end{eqnarray*}
In the above differential operator, 
second order terms cancel.
In fact, we can write the expression 
$-\partial_h\overline{\sigma}_{\ell h}
\mlt{\partial_j\sigma_{\ell k}}\partial_k
+\partial_h\mlt{\partial_j\overline{\sigma}_{\ell h}}
\sigma_{\ell k}\partial_k$, by an exchange of summation 
indexes $h,k$ in the second term, in the form 
\begin{align*}
&  \partial_k\sigma_{\ell h}
\mlt{\partial_j\overline{\sigma}_{\ell k}}\partial_h
-\partial_h\,\overline{\sigma}_{\ell h}
\mlt{\partial_j\sigma_{\ell k}} \partial_k 
 = \left( \sigma_{\ell h}
\mlt{\partial_j\overline{\sigma}_{\ell k}}
-\overline{\sigma}_{\ell h}
\mlt{\partial_j\sigma_{\ell k}}\right)\partial_h\partial_k \\
&    +
\left( \mlt{\partial_k\sigma_{\ell h}}
\mlt{\partial_j\overline{\sigma}_{\ell k}}
+\sigma_{\ell h}
\mlt{\partial_j\partial_k\overline{\sigma}_{\ell k}}\right)\partial_h
-\left( \mlt{\partial_h\overline{\sigma}_{\ell h}}
\mlt{\partial_j\sigma_{\ell k}}
+\overline{\sigma}_{\ell h}
\mlt{\partial_h\partial_j\sigma_{\ell k}}\right)\partial_k.
\end{align*}
This is a first order differential operator 
because both 
the second order coefficient vanishes by (\ref{4.2}) 
and 
$ 2\mathcal{L}(\partial_j) $ is equal to
\begin{align*}
&
\left( \mlt{\partial_k\sigma_{\ell h}}
\mlt{\partial_j\overline{\sigma}_{\ell k}}
+\sigma_{\ell h}
\mlt{\partial_j\partial_k\overline{\sigma}_{\ell k}}
+\mlt{\partial_j\overline{\sigma}_{\ell h}}\eta_\ell
-\overline{\sigma}_{\ell h} \mlt{\partial_j\eta_\ell}
\right)\partial_h\\
& -
\left( \mlt{\partial_h\overline{\sigma}_{\ell h}}
\mlt{\partial_j\sigma_{\ell k}}
+\overline{\sigma}_{\ell h}
\mlt{\partial_h\partial_j\sigma_{\ell k}}
+\mlt{\partial_j\overline{\eta}_{\ell}}\sigma_{\ell k}
-\overline{\eta}_\ell\mlt{\partial_j\sigma_{\ell k}}\right)\partial_k  \\
& + 
\mlt{\partial_j\partial_h\overline{\sigma}_{\ell h}}\eta_\ell
+ \mlt{\partial_j\overline{\sigma}_{\ell h}}
\mlt{\partial_h\eta_\ell}
- \hspace{-1pt}  \mlt{\partial_h\overline{\sigma}_{\ell h}}
\mlt{\partial_j\eta_\ell}
- \hspace{-1pt}  \overline{\sigma}_{\ell h}
\mlt{\partial_j\partial_h\eta_\ell}
\hspace{-1pt} +2i\Im\left(\overline{\eta}_\ell\mlt{\partial_j \eta_\ell}\right).
\end{align*}
Therefore $\mathcal{L}_0(\partial_j)=\nu_{jk}\partial_k + \xi_j $,
where
$
\nu_{jk} :=   
\Re\left( \mlt{ \partial_j\sigma_{\ell k}} \overline{\eta}_\ell
-\sigma_{\ell k} \mlt{ \partial_j \overline{\eta}_\ell } \right)
$ 
and
\begin{eqnarray*}
 2 \xi_j & := &
\mlt{ \partial_j\partial_h\overline{\sigma}_{\ell h} } \eta_\ell
+  \mlt{ \partial_j\overline{\sigma}_{\ell h} } \mlt{ \partial_h \eta_\ell }
-  \mlt{\partial_h\overline{\sigma}_{\ell h} } \mlt{ \partial_j \eta_\ell }
-\overline{\sigma}_{\ell h} \mlt{ \partial_j\partial_h\eta_\ell } 
\\
&&
+ 2 i \Im \left( \overline{\eta}_{\ell} \mlt{ \partial_j \eta_{\ell} } 
\right) .
\end{eqnarray*}
Since
$
\mathcal{L}_0(\Delta)  
=  
\partial_j \mathcal{L}_0(\partial_j) + \mathcal{L}_0(\partial_j)\partial_j + [\partial_j,L_\ell^*][\partial_j,L_\ell] 
$,
\begin{eqnarray*}
 \mathcal{L}_0(\Delta) 
 & = &
  \partial_j \left(  \nu_{jh}\partial_h + \xi_j \right)+  \left(  \nu_{jk}\partial_k + \xi_j \right)\partial_j 
  \\
 &&
 +  \left( \mlt{ \partial_j\overline{\eta}_\ell } -\partial_h \mlt{ \partial_j\overline{\sigma}_{\ell h} }
 \right)
\left(\mlt{ \partial_j\sigma_{\ell k} } \partial_k
+ \mlt{  \partial_j\eta_\ell } \right) .
\end{eqnarray*}
This gives 
\begin{align}
\label{4.4}
\mathcal{L}_0(\Delta)
 = \ &  2\nu_{jk}\partial_j\partial_k + \mlt{ \partial_j\nu_{jk} } \partial_k
          + 2 \xi_j\partial_j +  \mlt{ \partial_j\xi_j }
         - \mlt{ \partial_j\overline{\sigma}_{\ell h} } \mlt{ \partial_j\sigma_{\ell k} } \partial_h\partial_k 
         \\
         \nonumber
 &   - \left( \mlt{ \partial_h\partial_j\overline{\sigma}_{\ell h} } 
\mlt{ \partial_j\sigma_{\ell k} }
 + \mlt{ \partial_j\overline{\sigma}_{\ell h} } \mlt{ \partial_h\partial_j\sigma_{\ell k} } 
 \right) \partial_k 
 +  \mlt{\partial_j\overline{\eta}_\ell } \mlt{ \partial_j\sigma_{\ell k} } \partial_k
 \\
 \nonumber
&   - \mlt{ \partial_j\overline{\sigma}_{\ell h} } \mlt{ \partial_j\eta_\ell } \partial_h
- \mlt{ \partial_h\partial_j\overline{\sigma}_{\ell h} } \mlt{ \partial_j\eta_\ell }
- \mlt{ \partial_j\overline{\sigma}_{\ell h} } \mlt{ \partial_h\partial_j\eta_\ell }
\\
\nonumber
& + \mlt{ \partial_j\overline{\eta}_\ell } \mlt{ \partial_j\eta_\ell } .
\end{align}

By Condition H4.2, combining
$\mathcal{L}_0(C) =\mathcal{L}_0(-\Delta)+\mathcal{L}_0(|x|^2)$ 
with (\ref{4.3}) and (\ref{4.4})
we deduce that  $\mathcal{L}_0(C)$
is a second order differential operator of the form 
$\sum_{|\mu|\le 2} a_{\mu}\partial_\mu$ 
(in multiindex notation $\mu=(\mu_1,\dots,\mu_d), \  
|\mu|=\mu_1+\dots+\mu_d$, $\partial_\mu=\partial_{\mu_1}
\cdots\partial_{\mu_d}$)   with: 
$a_\mu$ bounded for $|\mu|=2$, $|a_\mu|\le K (t)(1+|x|)$ for $|\mu|=1$ and  $|a_\mu|\le K (t)(1+|x|^2)$ for $|\mu|=0$.
The conclusion follows them from applications of 
Lemma \ref{lemma:key-op-ineq} with some long 
but straightforward computations. 
\end{proof}

\begin{lemma}\label{lem:estimate-L0}
Under the Hypothesis \ref{hyp:growth-cond},
Condition H2.3 holds.
\end{lemma}

\begin{proof} 
Since 
$
\left[ C, L_\ell\right] 
= 
   \left( (-\partial_j + x_j)\left[\partial_j + x_j,L_\ell \right]   + \left[-\partial_j + x_j,L_\ell\right] (-\partial_j + x_j) \right) 
$,
\begin{align*}
 \left[ C, L_\ell\right]
 = \ &
 \left( -\partial_j + x_j \right)
\left( \mlt{  \partial_j \sigma_{\ell k} } \partial_k 
+ \mlt{ \partial_j \eta_\ell } 
- \delta_{jk}\sigma_{\ell k}\right) 
 \\
 & +  \left(- \mlt{  \partial_j \sigma_{\ell k} } \partial_k 
- \mlt{ \partial_j \eta_\ell } - \delta_{jk}\sigma_{\ell k}\right) (-\partial_j + x_j),
\end{align*}
where $ \delta_{jk}$ is the Kronecker delta.
The term with two partial derivatives writes as
\[
- \partial_j \mlt{  \partial_j \sigma_{\ell k} } \partial_k 
+ \mlt{  \partial_j \sigma_{\ell k} } \partial_k\partial_j
=  2 \mlt{  \partial_j \sigma_{\ell k} } \partial_k\partial_j
- \mlt{  \partial_j^2\sigma_{\ell k} } \partial_k ,
\]
terms with a single partial derivative are 
\begin{eqnarray*} 
& & -\partial_j  \mlt{  \partial_j \eta_\ell }
+\delta_{jk}\partial_j\sigma_{\ell k} 
+ x_j \mlt{  \partial_j \sigma_{\ell k}} \partial_k 
+\mlt{  \partial_j \eta_\ell } \partial_j 
+\delta_{jk}\sigma_{\ell k}\partial_j 
- \mlt{ \partial_j \sigma_{\ell k}} \partial_k x_j  
\\
& = &  - \mlt{  \Delta\eta_\ell } 
+ 2 \delta_{jk} \sigma_{\ell k}\partial_j ,
\end{eqnarray*}
and terms with no partial derivatives sum up  $-2\delta_{jk}x_j\sigma_{\ell k} $.
Therefore
\[
\left[ C, L_\ell\right] 
 =  2  \mlt{ \partial_j \sigma_{\ell k} } \partial_j \partial_k 
 -  \mlt{ \partial_j ^2\sigma_{\ell k} } \partial_k 
 + 2 \sigma_{\ell j}\partial_j
 -2x_j\sigma_{\ell j}           
 - \mlt{ \Delta\eta_\ell } .
\]
We now use Condition H4.2 and Lemma \ref{lemma:key-op-ineq},
together with straightforward inequalities and estimates,
 to obtain 
$\left\|[C,L_\ell(t)]f\right\|^2
\leq
K(t) \left( \left\| Cf\right\|^2 + \left\| f\right\|^2 \right)
$.
Thus,
the claimed inequality in Condition H2.3 follows from 
Lemmata \ref{lem:comm-H-C-estim}, \ref{lemma:CCP-part-C}    
and \ref{lem:estimate-L0-step1}.
\end{proof}

\begin{proof}[Proof of Theorem \ref{teo:Smooth}]
Hypothesis 1 and Conditions H2.1, H2.2 of Hypothesis 2 hold 
by Lemma \ref{lem:G-satisfies-H.2.1}. 
In Lemma \ref{lem:estimate-L0} we verify  Condition H2.3.  
According the definition of $G \left(t \right)$
we have $
 2\Re\left\langle f, G \left( t \right) f\right\rangle 
+ \sum_{\ell=1}^{\infty }\left\Vert L_\ell \left( t \right) f\right\Vert ^{2} = 0
$
for all $f \in C^\infty_c(\mathbb{R}^d;\mathbb{C})$.
Therefore, 
Condition H3.1
(stronger form of H2.4) holds,
because  $C^\infty_c(\mathbb{R}^d;\mathbb{C})$ is a 
core for $C$ and the operators $G(t),L_\ell(t)$ are relatively 
bounded with respect to $C$ with bound uniform for $t$ in bounded intervals $[0,T]$.
Hence,
applying Theorems \ref{teorema1},  \ref{teorema2} and \ref{teorema7}
we get the assertions of the theorem. 
\end{proof}

\subsection{Proof of Theorem \ref{th:Ehrenfest}}
\label{subsec:teoEhrenfest}

\begin{proof}[Proof of Theorem \ref{th:Ehrenfest}]

Consider the stopping time 
$$
\tau _{m} = \inf{ \left\{  t \geq0: \left\Vert  X_{t}\left( \xi \right)  \right\Vert > m \right\} } \wedge T ,
$$
where $T > 0$ and $m \in \mathbb{N}$.
For any $n \in \mathbb{N}$,
we set  $ A_n =  R_n A R_n $,
with
$ R_n = n \left( n + C \right)^{-1}$.
From
$A \in  \mathfrak{L}\left( \left(  \mathcal{D}\left( C \right),  \left\Vert  \cdot \right\Vert_C\right),  \mathfrak{h} \right)$
it follows that $A_n \in \mathfrak{L}\left( \mathfrak{h} \right)$,
and so using  the complex It\^o formula we obtain
\begin{equation}
 \label{eq:5}
  \left\langle   X_{t \wedge \tau _{m}} \left( \xi \right) ,  A_n X_{t \wedge \tau _{m}} \left( \xi \right) \right\rangle
 =
 \left\langle  \xi  , A_n \xi  \right\rangle 
 +  \int_0^{t \wedge \tau _{m}} 
 \mathcal{L} \left( s, A_n,  X_{s} \left( \xi \right) \right)
 ds
 + M_{t \wedge \tau _{m}} ,
\end{equation}
where 
$t \in \left[ 0, T \right]$,
$$
M_t
=
\sum_{\ell = 1}^{\infty}
\int_0^{t} \left( 
\left\langle   X_{s} \left( \xi \right) , A_n L_{\ell}  \left( s \right) X_{s} \left( \xi \right) \right\rangle
+
\left\langle  L_{\ell}  \left( s \right)  X_{s} \left( \xi \right) , A_n X_{s} \left( \xi \right) \right\rangle
 \right) dW_s^{\ell} ,
$$
and for all $x \in \mathcal{D}\left( C \right)$,
$$
 \mathcal{L} \left( s, A_n, x  \right)
=
\left\langle  x , A_n G  \left( s \right) x \right\rangle
 +
 \left\langle  G  \left( s \right) x , A_n  x \right\rangle
 +
 \sum_{\ell = 1}^{\infty} \left\langle  L_{\ell}   \left( s \right) x , A_n   L_{\ell}   \left( s \right)  x  \right\rangle .
$$

Combining Lemma \ref{lema5} with the Cauchy-Schwarz inequality we get
\begin{align*}
 &
 \mathbb{E} \sum_{\ell = 1}^{\infty}
\int_0^{t  \wedge \tau _{m} } \left|
\left\langle   X_{s} \left( \xi \right) , A_n L_{\ell}  \left( s \right) X_{s} \left( \xi \right) \right\rangle
+
\left\langle  L_{\ell}  \left( s \right)  X_{s} \left( \xi \right) , A_n X_{s} \left( \xi \right) \right\rangle
 \right|^2 ds
 \\
 & \leq 8 m^3 \left\| A_n \right\|^2 \int_0^{T}   \mathbb{E} \left\|  G  \left( s \right)  X_{s} \left( \xi \right)  \right\|  ds ,
\end{align*}
which together with Condition H2.1, yields 
$
\mathbb{E}  M_{t \wedge \tau _{m}}  = 0
$.
Then, 
(\ref{eq:5}) leads to
\begin{equation}
\label{eq:6}
 \mathbb{E}  \left\langle   X_{t \wedge \tau _{m}} \left( \xi \right) ,  A_n X_{t \wedge \tau _{m}} \left( \xi \right) \right\rangle
 =
\mathbb{E}   \left\langle  \xi  , A_n \xi  \right\rangle 
 + \mathbb{E}  \int_0^{t \wedge \tau _{m}} 
 \mathcal{L} \left( s, A_n,  X_{s} \left( \xi \right) \right)
 ds .
\end{equation}
Since 
$
\mathbb{E} \sup_{s \in \left[0, T \right]} \left\|  X_{s} \left( \xi \right)  \right\| ^2 < + \infty
$,
applying the dominated convergence theorem gives
\[
\lim_{m \rightarrow \infty}
\mathbb{E}  \left\langle   X_{t \wedge \tau _{m}} \left( \xi \right) ,  A_n X_{t \wedge \tau _{m}} \left( \xi \right) \right\rangle
=
\mathbb{E}  \left\langle   X_{t} \left( \xi \right) ,  A_n X_{t} \left( \xi \right) \right\rangle .
\]
Letting $m \rightarrow \infty$ in (\ref{eq:6})  we deduce,
using the dominated convergence theorem, that
$$
\mathbb{E}  \left\langle   X_{t} \left( \xi \right) ,  A_n X_{t} \left( \xi \right) \right\rangle
 = 
\mathbb{E}   \left\langle  \xi  , A_n \xi  \right\rangle 
 + \mathbb{E}  \int_0^{t} 
 \mathcal{L} \left( s, A_n,  X_{s} \left( \xi \right) \right)
 ds,
$$
and so from Fubini's theorem we obtain
\begin{equation}
 \label{eq:8}
 \mathbb{E}  \left\langle   X_{t} \left( \xi \right) ,  A_n X_{t} \left( \xi \right) \right\rangle
 =
 \mathbb{E}   \left\langle  \xi  , A_n \xi  \right\rangle 
 + \int_0^{t}  \mathbb{E}   \mathcal{L} \left( s, A_n,  X_{s} \left( \xi \right) \right) ds .
\end{equation}

Let  $x \in \mathcal{D}\left( C \right)$.
By Conditions H2.2 and H5.1,
analysis similar to that in the proof of Lemma \ref{lema2i} shows that
$L_{\ell}   \left( s \right) x \in  \mathcal{D}\left( C^{1/2} \right)$ and 
\begin{equation}
\label{eq:11}
  \sum_{\ell = 1}^{\infty}
  \left\|  C^{1/2} L_{\ell}   \left( s \right) x  \right\|^2 
 \leq K \left( s \right)  \left\Vert x \right\Vert_{C}^{2} .
\end{equation}
Since $R_n C  \subset C R_n$,
$C^{1/2}$ commutes with $R_n$,
and so Condition H5.2 leads to
\begin{eqnarray*}
&& \left\|   B_j R_n L_{\ell}   \left( s \right) x - B_j L_{\ell}   \left( s \right) x  \right\|^2
\\
&& \leq
K \left( 
\left\|   R_n C^{1/2} L_{\ell}   \left( s \right) x - C^{1/2} L_{\ell}   \left( s \right) x  \right\|^2
+
\left\|   R_n  L_{\ell}   \left( s \right) x -  L_{\ell}   \left( s \right) x  \right\|^2 
\right) ,
\end{eqnarray*}
with $j=1,2$.
This  implies  
\begin{equation}
 \label{eq:7}
 B_j R_n L_{\ell}   \left( s \right) x 
\longrightarrow_{n \rightarrow \infty}
 B_j L_{\ell}   \left( s \right) x .
\end{equation}
Moreover, 
using  $R_n C^{1/2}  \subset C^{1/2} R_n$ we deduce that
\begin{eqnarray*}
   \left\|   B_j R_n L_{\ell}   \left( s \right) x  \right\|^2
& \leq &
 K 
 \left(  \left\|  R_n C^{1/2} L_{\ell}   \left( s \right) x  \right\|^2 +  \left\|   R_n L_{\ell}   \left( s \right) x  \right\|^2 \right)
 \\
 & \leq &
 K  
 \left(  \left\|  C^{1/2} L_{\ell}   \left( s \right) x  \right\|^2 +  \left\|  L_{\ell}   \left( s \right) x  \right\|^2 \right) .
\end{eqnarray*}

Lemma \ref{lema5} and Condition H2.1 lead to
$
  \sum_{\ell = 1}^{\infty}
  \left\|  L_{\ell}   \left( s \right) x  \right\|^2 
 \leq K \left( s \right)  \left\Vert x \right\Vert_{C}^{2} 
$.
Therefore,
applying the dominated convergence theorem, together with (\ref{eq:11}) and (\ref{eq:7}), yields
\begin{eqnarray*}
&&
 \int_0^{t}  \mathbb{E}  
 \sum_{\ell = 1}^{\infty} \left\langle B_1 R_n L_{\ell}   \left( s \right) X_{s} \left( \xi \right) , 
 B_2 R_n   L_{\ell}   \left( s \right)  X_{s} \left( \xi \right)  \right\rangle
 ds
 \\
 && \hspace{3cm} \longrightarrow_{n \rightarrow \infty} 
 \int_0^{t}  \mathbb{E}  
 \sum_{\ell = 1}^{\infty} \left\langle B_1  L_{\ell}   \left( s \right) X_{s} \left( \xi \right) , 
 B_2    L_{\ell}   \left( s \right)  X_{s} \left( \xi \right)  \right\rangle
 ds .
 \end{eqnarray*}
Hence
\begin{eqnarray}
\label{eq:9}
&&
 \int_0^{t}  \mathbb{E}  
 \sum_{\ell = 1}^{\infty} \left\langle B_1 R_n L_{\ell}   \left( s \right) X_{s} \left( \xi \right) , 
 B_2 R_n   L_{\ell}   \left( s \right)  X_{s} \left( \xi \right)  \right\rangle
 ds
\\
\nonumber
 &&  \hspace{3cm} \longrightarrow_{n \rightarrow \infty} 
  \int_0^{t}  \sum_{\ell = 1}^{\infty}  \mathbb{E}  
  \left\langle B_1  L_{\ell}   \left( s \right) X_{s} \left( \xi \right) , 
 B_2    L_{\ell}   \left( s \right)  X_{s} \left( \xi \right)  \right\rangle
 ds .
\end{eqnarray}
According to $R_n ^* = R_n$,
for any $x \in \mathcal{D}\left( C \right)$  we have
\begin{eqnarray*}
  \mathcal{L} \left( s, A_n, x  \right)
 & = &
\left\langle R_n A^* R_n x ,  G  \left( s \right) x \right\rangle
 +
 \left\langle  G  \left( s \right) x , R_n A R_n x \right\rangle
 \\
 &&
+  \sum_{\ell = 1}^{\infty} \left\langle B_1 R_n L_{\ell}   \left( s \right) x , B_2 R_n   L_{\ell}   \left( s \right)  x  \right\rangle .
\end{eqnarray*}
By (\ref{eq:9}) and Condition H5.3,
letting $n \rightarrow \infty$ in  (\ref{eq:8}) we get (\ref{eq:10}).
\end{proof}

\subsection{Proof of Theorem \ref{th:Ehrenfest_Model}}
\label{subsec:teoEhrenfest_Model}

\begin{proof}[Proof of Theorem \ref{th:Ehrenfest_Model}]
Let $C = - \Delta+ \left| x \right| ^{2}$.
According to Theorem \ref{teo:Smooth},
(\ref{eq:SSE}) has a unique strong $C$-solution with initial datum in
 $L_{C}^{2}\left(\mathbb{P};\mathfrak{h}\right)$.
 Moreover,
 in the proof of Theorem \ref{teo:Smooth}
we verify that $C$ satisfies Hypothesis \ref{hyp:CF-inequality}.

Suppose that $f$ belongs to $C^\infty_c(\mathbb{R}^d, \mathbb{C})$,
which is a core for $C$.
Then, 
for any $\ell = 1, \ldots, m$ and $t \geq 0$ we have
$$
 \left\|   C^{1/2} L_{\ell}   \left( t \right) f  \right\|^2
 =
 \sum_{j=1}^{d} 
 \left\| i \partial_j \left(  L_{\ell}   \left( t \right) f \right) \right\|^2
 +
 \left\|  \left| x \right| L_{\ell}   \left( t \right) f  \right\|^2 .
$$
Since 
\begin{equation}
\label{eq:12}
 \sum_{j,k=1}^d  \left\|  \partial_j \partial_k  f  \right\|^2
=
\left\| - \Delta  f  \right\|^2 ,
\end{equation}
combining  Hypothesis \ref{hyp:growth-cond} with Lemma \ref{lemma:key-op-ineq} 
yields
Condition H5.1.

Consider  $f \in C^\infty_c(\mathbb{R}^d, \mathbb{C})$.
Then
$
 \left\| \mlt{c_j}   f  \right\|^2
 \leq 
 K \left( \left\|  f  \right\|^2 + \left\langle  f ,  \left|  x  \right|^2 f \right\rangle
 \right)
 $
 and
 $$
 \left\| \mlt{b_j} \partial_k   f  \right\|^2
 \leq 
 K 
 \left\langle  f ,  - \partial^{2}_{k}  f \right\rangle .
$$
In addition,
$
\left\|  \partial_k \mlt{a_j}   f  \right\|^2
\leq
2 \left\|  \mlt{ \partial_k a_j}   f  \right\|^2 + 2  \left\|  \mlt{ a_j}    \partial_k f  \right\|^2
$.
Therefore 
$$
\left\| B_j   f  \right\|^2 
\leq 
K \left( \left\|  f  \right\|^2 + \left\langle  f ,  C  f \right\rangle \right)
=
K  \left\|  f  \right\|^{2}_{C^{1/2}} ,
$$
and so $B_j$ satisfies  Condition H5.2, 
because $C^\infty_c(\mathbb{R}^d, \mathbb{C})$ is a core for $C^{1/2}$.

We now take 
$B_1 =  \mlt{b_1} \partial_\ell$ and $B_2 =  \partial_k \mlt{a_2}$,
and so $ A = -  \partial_\ell  \mlt{\bar{b_1}}  \partial_k \mlt{a_2} $.
For any  $f \in C^\infty_c(\mathbb{R}^d, \mathbb{C})$,
\begin{eqnarray*}
 A f 
& = &
- \left(  \partial_\ell  \bar{b_1} \right) \left( \partial_k a_2 \right) f
-  \bar{b_1} \left(  \partial_\ell   \partial_k a_2 \right) f
-  \bar{b_1} \left(  \partial_k a_2 \right) \partial_\ell f
- \left(  \partial_\ell  \bar{b_1} \right)  a_2 \partial_k f
\\
&&
-  \bar{b_1} \left(  \partial_\ell  a_2 \right) \partial_k f
-  \bar{b_1} a_2 \partial_\ell  \partial_k  f .
\end{eqnarray*}
Using Lemma \ref{lemma:key-op-ineq}, together with (\ref{eq:12}),  yields 
$
\left\| A   f  \right\|^2 \leq K \left\|   f  \right\|^2_{C} 
$,
and hence for all $f \in  \mathcal{D}\left( C \right)$,
$
\left\| A   f  \right\|^2 \leq K \left\|   f  \right\|^2_{C}
$
since $C^\infty_c(\mathbb{R}^d, \mathbb{C})$ is a core for $C$.
Similarly, 
we obtain 
$
\left\| A^{*}   f  \right\|^2 \leq K \left\|   f  \right\|^2_{C}
$
for all  $f \in  \mathcal{D}\left( C \right)$.
Thus Condition H5.3 holds in this case.
In the same manner we can check Condition H5.3
for the other possible choices of $B_1$ and $B_2$.
Finally,
applying Theorems \ref{th:Ehrenfest}  and  \ref{cor:Ehrenfest} we get (\ref{eq:10}) and (\ref{eq:14}),
respectively.
\end{proof}

\subsection{Proof of Corollary \ref{cor:heat}}
\label{subsec:cor_heat}

\begin{proof}[Proof of Corollary \ref{cor:heat}]
Set $P =  - i d/dx $.
Suppose that either $A = P^2$ or $A = \mlt{V}$.
From $L_1^{*} = - L_1$ it follows that
for all $f \in C^\infty_c(\mathbb{R}, \mathbb{C})$,
$$
\left\langle A^* f , G f \right\rangle +  \left\langle G f, A f \right\rangle 
+ \left\langle \sqrt{A} L_{1} f , \sqrt{A} L_{1} f \right\rangle
=
\left\langle f , 
\left( -i \left[ A, H \right] - \frac{1}{2} \left[ \left[ L_1 , A \right] , L_1\right] \right) f 
\right\rangle .
$$
Using $ \left[\,\mlt{V} , P \,\right] = i \mlt{V^{\prime}}$ yields
$$
-i \left[ A, H \right] - \frac{1}{2} \left[ \left[ L_1 , A \right] , L_1\right]
=
\left\{
\begin{array}{ll}
\frac{1}{2M} \left( \mlt{V^\prime} P + P \mlt{V^\prime} \right),
&
\hbox{\rm if } \  A =  \mlt{V}
\\
 -  \left( \mlt{V^\prime} P + P \mlt{V^\prime} \right) + \eta^2,
 &
\hbox{\rm if } \ A = P^2
\end{array}
\right. .
$$  
Since  $A$, $G$, $\sqrt{A} L_{1}$, $\mlt{V^\prime} P$ and $P \mlt{V^\prime}$ 
are relatively bounded with respect to $C= -  d^2/dx^2  + \mlt{x^2}$,
for all $f \in   \mathcal{D}\left( C \right)$ we have 
\begin{eqnarray}
\label{eq:13}
 && \left\langle A^* f , G f \right\rangle +  \left\langle G f, A f \right\rangle 
 + \left\langle \sqrt{A} L_{1} f , \sqrt{A} L_{1} f \right\rangle
 \\
 \nonumber
 && \hspace{2cm} =
 \left\{
\begin{array}{ll}
\frac{1}{2M} \left\langle f,  \ \left( \mlt{V^\prime} P + P \mlt{V^\prime} \right) f \right\rangle,
&
if \  A =  \mlt{V}
\\
 -  \left\langle f,    \left( \mlt{V^\prime} P + P \mlt{V^\prime} \right)  f \right\rangle 
+ \left\langle f,  \eta^2 f \right\rangle, 
&
 if \ A = P^2
\end{array}
\right. ,
\end{eqnarray}
because  $C^\infty_c(\mathbb{R}, \mathbb{C})$ is a core for $C$.
Combining (\ref{eq:13}) with Theorem \ref{th:Ehrenfest_Model} we obtain
\begin{equation}
 \label{eq:15}
 \mathbb{E} \left\langle X_{t} ,\mlt{V} X_{t}  \right\rangle
 =  
\mathbb{E} \left\langle  \xi  , \mlt{V}  \xi  \right\rangle
+
  \frac{ 1}{2M}  \int_0^t 
   \mathbb{E} \left\langle X_{s} , \left( \mlt{V^\prime} P + P \mlt{V^\prime} \right)  X_{s}  \right\rangle
ds
\end{equation}
and
\begin{align}
\nonumber
  \mathbb{E} \left\langle X_{t} , \frac{1}{2M} P^2 X_{t}  \right\rangle
 = &  \
\mathbb{E} \left\langle  \xi  ,  \frac{1}{2M}  P^2  \xi  \right\rangle
-
  \frac{1}{2M}   \int_0^t 
   \mathbb{E} \left\langle X_{s} ,   \left( \mlt{V^\prime} P + P \mlt{V^\prime} \right) X_{s}  \right\rangle
ds
\\
\label{eq:16}
&
+ \frac{\eta^2}{2M}  t,
\end{align}
where we abbreviate $X_{t} \left( \xi \right)$ to $X_t$.
Adding (\ref{eq:15}) and (\ref{eq:16}) gives (\ref{eq:17}).
\end{proof}

\appendix
\section{}

\subsection{Proof of Theorem \ref{teorema2}}
\label{subsec:teorema2}

\begin{proof}[Proof of Theorem \ref{teorema2}]
Define 
$
\tau _{n} = \inf{ \left\{  t \geq0: \left\Vert  X_{t}\left( \xi \right)  \right\Vert > n \right\} } \wedge T
$,
where $T$ is a given positive real number and  $n \in \mathbb{N}$. 
Combining Condition H3.1 with It\^{o}'s formula we obtain
\begin{equation}
\label{6.1}
\left\Vert  X_{t \wedge \tau _{n}}\left( \xi \right)  \right\Vert ^{2}  
= \left\Vert  \xi \right\Vert ^{2}  
+ \sum_{\ell=1}^{\infty} \int_{0}^{t \wedge \tau _{n}} 
2 \Re \left\langle  X_{s}\left( \xi \right) , L_{\ell}  \left( s \right)  X_{s}\left( \xi \right)\right\rangle dW_{s}^{\ell}.
\end{equation}
Conditions H2.1 and H3.1 yield
\[
\sum_{\ell=1}^{\infty}  \mathbb{E}  \int_{0}^{ \tau _{n}} \left( \Re \left\langle  X_{s}\left( \xi \right) , L_{\ell} \left( s \right)   X_{s}\left( \xi \right)\right\rangle \right)^{2} ds
\leq K_{n,T} \left(1+   \mathbb{E}  \left\Vert   \xi  \right\Vert_{C}^{2} \right),
\]
where $K_{n,T}$ is a constant depending of $n$ and $T$,
hence (\ref{6.1}) shows that $  \left\Vert  X^{\tau _{n}}\left( \xi \right)  \right\Vert ^{2} $ is a martingale. 
We now use  Fatou's lemma to deduce the supermartingale property of $ \left(  \left\Vert  X_{t}\left( \xi \right)  \right\Vert ^{2}\right) _{t \in \left[0, T\right]} $. 

Since 
$
 \mathbb{E} \left( \sup_{s \in \left[0, T\right]}  \left\Vert  X_s\left( \xi \right)  \right\Vert ^{2} \right)
 < \infty
$
(see, e.g., Th. 4.2.5 of  \cite{Prevot}),
applying the dominated convergence theorem gives
\[
\mathbb{E}  \left\Vert  X_t \left( \xi \right)  \right\Vert ^{2} 
=
\lim_{n \rightarrow \infty} \mathbb{E}  \left\Vert  X_{t\wedge \tau _{n}} \left( \xi \right)  \right\Vert ^{2} 
=
\lim_{n \rightarrow \infty} \mathbb{E}  \left\Vert   \xi   \right\Vert ^{2} .
\]
Therefore the supermartingale
$ \left(  \left\Vert  X_{t}\left( \xi \right)  \right\Vert ^{2}\right) _{t \in \left[0, T\right]} $
is in fact a martingale. 
\end{proof}

\subsection{Proof of Lemma \ref{lema30}}
\label{subsec:lemma30}

\begin{proof}[Proof of Lemma \ref{lema30}]
 
Let $\left( \chi _{j}\right)_{j \in \mathbb{N}}$ be an orthonormal basis of 
$$L^{2}\left(\left( \Omega,\mathfrak{G}_{T}^{\xi ,W},\mathbb{P} \right), \mathfrak{h}\right). $$ 
Combining the Cauchy-Schwarz inequality with  (\ref{3.36})  we obtain the equicontinuity of the family of complex functions $\left( \mathbb{E} \left\langle \chi _{j} , X^{n}  \right\rangle  \right)_{n \in \mathbb{N}}$, with $j \in \mathbb{N}$. Using Lemma \ref{lema1}, the Arzel\`a-Ascoli theorem and diagonalization arguments we deduce that   can extract from any subsequence of  
$\left( X^{n} \right)_{n \in \mathbb{N}}$ a subsequence 
$\left( X^{n_{k}}    \right)_{k \in \mathbb{N}}$ such that $ \mathbb{E} \left\langle \chi _{j} , X^{n_{k}}  \right\rangle$ is uniformly convergent in $\left[  0,T\right]$ for any $j \in \mathbb{N}$. Lemma \ref{lema1} now shows that 
$ X_{t}^{n_{k}}$ is weakly convergent in $L^{2}\left(\left( \Omega,\mathfrak{G}_{T}^{\xi ,W},\mathbb{P} \right), \mathfrak{h}\right)$ for any $t \in \left[  0,T\right]$. 
Since $X^{n_{k}}_{t} $  is 
$\mathfrak{G}_{t}^{\xi ,W}$-measurable, for any  $t \in \left[  0,T\right]$ there exists a  $\mathfrak{G}_{t}^{\xi ,W}$-measurable random variable $\psi_{t}$ satisfying 
\begin{equation}
\label{3.38}
X_{t}^{n_{k}}  \longrightarrow_{k\rightarrow\infty} \psi_{t}  \qquad weakly\ in\ L^{2}\left(
\left( \Omega,\mathfrak{G}_{t}^{\xi ,W},\mathbb{P} \right), \mathfrak{h}\right).
\end{equation}

Assume that $\left( e_{j} \right)_{j \in \mathbb{N}}$ is an 
orthonormal basis of $\mathfrak{h}$. According to (\ref{3.38}) 
we have 
\[
\left\langle e_{j}, X_{t}^{n_{k}} \right\rangle  \longrightarrow_{k\rightarrow\infty} \left\langle e_{j}, \psi_{t} \right\rangle  \quad weakly\ in\ L^{2}\left(
\left( \Omega,\mathfrak{G}_{t}^{\xi ,W},\mathbb{P} \right), \mathbb{C}\right) .
\]
Thus,
from (\ref{3.36}) it follows that
\[
\mathbb{E} \left\vert \left\langle e_{j}, \psi_{t} -  \psi_{s} \right\rangle  \right\vert ^{2}
\leq 
 \liminf_{k \rightarrow \infty} \mathbb{E} \left\vert \left\langle e_{j}, X_{t}^{n_{k}} - X_{t}^{n_{k}} \right\rangle  \right\vert ^{2}\\
 \leq   K_{T, \xi} \left( t - s \right).
\]
It follows that $ \left\langle e_{j}, \psi \right\rangle $ has a  $\left(\mathfrak{G}_{t+}^{\xi,W}\right)_{t\in\left[  0,T\right]}$-predictable version 
$\widetilde{  \left\langle e_{j}, \psi \right\rangle} $ (see, e.g., Proposition 3.6 of  \cite{DaPrato}). We define $\mathfrak{a}$ to be the set of all $\left( t, \omega \right)$ belonging to $\left[ 0, T\right] \times \Omega$ such that $\sum_{j=1}^{n} \widetilde{  \left\langle e_{j}, \psi \right\rangle} _{t} \left( \omega \right) e_{j}$ converge as $n$ goes to $\infty$. The proof is completed by choosing 
$
Z_{t}   \left(  \omega\right) 
=
\sum_{j=1}^{\infty} \widetilde{  \left\langle e_{j}, \psi \right\rangle} _{t} \left( \omega \right) e_{j}
$ 
if $\left(  t,\omega\right)  \in \mathfrak{a}$,
and 
$Z_{t}  \left(  \omega\right) = 0 $
provided that $\left(  t,\omega\right)  \notin \mathfrak{a}$.
Thus $Z$ becomes a version of $\psi$.
\end{proof}

\section*{Acknowledgements}
The authors wish to thank an anonymous referee for constructive comments, in particular 
those that led to Lemma 2.4.


\providecommand{\noopsort}[1]{}\providecommand{\singleletter}[1]{#1}%

\end{document}